\documentclass[]{article}
\usepackage{jheppub}
\usepackage{amsmath}
\usepackage[english]{babel}
\usepackage{amssymb}
\usepackage{mathtools}
\usepackage{color}
\definecolor{ao}{rgb}{0.0, 0.5, 0.0}

\usepackage[caption = false]{subfig}
\usepackage{graphicx}
\usepackage{longtable}
\usepackage{tikz}
\usepackage{comment}
\usepackage{capt-of}
\usepackage{amsthm}
\usepackage{hyperref}
\hypersetup{
    colorlinks=true,
    linkcolor=blue
    }

\usepackage{bm}

\newtheorem{remark}{Remark}
\newtheorem{theorem}{Theorem}
\newtheorem{proposition}{Proposition}
\newtheorem{corollary}{Corollary}
\newtheorem{definition}{Definition}
\newtheorem{lemma}{Lemma}

\def\0{\mbox{\tiny $0$}}
\def\1{\mbox{\tiny $1$}}
\def\2{\mbox{\tiny $2$}}
\def\3{\mbox{\tiny $3$}}
\def\4{\mbox{\tiny $4$}}
\def\5{\mbox{\tiny $5$}}
\def\6{\mbox{\tiny $6$}}
\def\7{\mbox{\tiny $7$}}
\def\8{\mbox{\tiny $8$}}
\def\9{\mbox{\tiny $9$}}

\def\k{k_{_{B}}}

\def\r{\rangle}
\def\l{\langle}
\def\m{\bar{m}}

\def\q{\bar{q}}

\def\qq{q_{12}}

\def\b{\beta^{'}}

\newcommand{\SOMMA}[2]{\displaystyle\sum\limits_{#1}^{#2}}
\newcommand{\sommaSigma}[1]{\displaystyle\sum\limits_{\lbrace#1\rbrace}}

\newcommand{\pder}[2]{\ensuremath{\dfrac{\partial #1}{\partial #2}}}

\long\def \beq#1\eeq {\begin{equation} #1 \end{equation}}
\long\def \beaq#1\eeaq {\begin{equation}\begin{aligned} #1 \end{aligned}\end{equation}}
\long\def \bes#1\ees {\begin{equation}\begin{split} #1 \end{split} \end{equation}}
\long\def \bea#1\eea {\begin{eqnarray} #1 \end{eqnarray}}
\long\def \bse[#1]#2\ese {\begin{subequations}\label{#1}\begin{align} #2 \end{align}\end{subequations}}

\newcommand{\dt}{\frac{\partial}{\partial t}}

\graphicspath{{Plots/}} 

\title{Rigorous approaches for spin glass and Gaussian spin glass with P-wise interactions}
\author[a, b, c]{Linda Albanese,}
\author[a, b]{Andrea Alessandrelli}

\affiliation[a]{Dipartimento di Matematica e Fisica {\em Ennio De Giorgi},  Universit\`a  del Salento, Via per Arnesano, 73100, Lecce, Italy}
\affiliation[b]{Scuola Superiore ISUFI, Campus Ecotekne, Via Monteroni, 73100, Lecce,Italy}
\affiliation[c]{Istituto Nazionale di Fisica Nucleare, Campus Ecotekne, Via Monteroni, 73100, Lecce,Italy}

\abstract{Purpose of this paper is to face up to P-spin glass and Gaussian P-spin model, i.e. spin glasses with polynomial interactions of degree $P > 2$. We consider the replica symmetry and first step of replica simmetry breaking assumptions and we solve the models via transport equation and Guerra's interpolating technique, showing that we reach the same results. 

Thus, using rigorous approaches, we recover the same expression for quenched statistical pressure and self-consistency equation in both assumption found with other techniques, including the well-known \textit{replica trick} technique. 

At the end, we show that for $P=2$ the Gaussian P-spin glass model is intrinsecally RS. 
}

\begin{document}

\maketitle

\section*{Introduction}
In last few decades statistical mechanics has played an important role in different fields, e.g. spin glasses\cite{MPV, Bovier1} and neural networks\cite{Amit, PRLNN, leonelli}. For this reason there has been a lot of research to tackle increasingly complex models through rigorous mathematical technique\cite{Bovier2}, alternative to the well-known replica-trick\cite{MPV}. Furthermore the growing interest in Replica Symmetry Breaking and its connection with ultrametricity, as just proved in \cite{PanchenkoUltra, GG}, has taken the Scientific Community to take into account the results linked to this assumption\cite{Tala_break, Zecchina1, Remi1, Zecchina2}. 

\par\medskip
In this paper we focus on mathematical methods on spin glasses, considered a real \textit{challenge for mathematicians}\cite{tala}.
In particular, in this paper we want to generalize some computations for Sherrington Kirkpatrick  and Gaussian models with two-spin interactions in a polynomial interactions of degree $P > 2$, with two different rigorous mathematical approaches, namely Guerra's interpolating technique, introduced in \cite{GuerraSum, Guerra}, and transport equation, introduced in Replica Symmetric assumption in \cite{AABF-NN2020} and in Replica Symmetry Breaking \cite{lindaRSB} for the description of Hopfield model and, afterwards, for deep Hopfield networks \cite{trans_DHN}. Thus we prove that both methods are consistent and mathematically valid for these models. 

Furthermore we face to RS and 1RSB assumptions, showing that, as fas as Gaussian spin glass model concernes, we reach the results in \cite{AdrianoGauss, glassy} for $P=2$ and in \cite{crisanti, gardner, Tala3} for $P >2$. For the 1RSB assumption, we use the broken replica interpolation introduced by Francesco Guerra for the Sherrington Kirkpatrick model \cite{Guerra} and, moreover, we adapt it into a PDE framework (standard transport equation). 

We stress that the present paper is focused on mathematical methods to deepen these models. So, we do not face up to physical interpretation of our results, for which we remind to \cite{crisanti, GrossMezard, Tala3}. 

\par \medskip
The paper is structured as follows.

In the first section we analyze P-spin glass model, namely Sherrington Kirkpatrick with P-wise interactions, starting from the main characters until the expression in RS and 1RSB assumption of quenched statistical pressure and self-consistency equations. For both assumption we compute via Guerra's interpolating technique and transport equation, showing that we reach the same results.
We proceed in similar way for Gaussian P-spin glass in the second section, coming to the emergent property for this model too. 
The paper is closed by a section concerning the case $P=2$ for Gaussian P-spin glass, in which we show that it is intrinsecally RS, and some appendices with computations from results of the paper.

\section{P-spin glass model}
In this section we deal with P-spin glass model, namely Sherrington-Kirkpatrick (SK) model with $P$-wise interactions, introduced by Derrida in \cite{Derrida} and meticulously described afterwards, e.g. \cite{gardner, Tala4, Tala5, GrossMezard, MPV}. In particular we find in RS and 1RSB assumption the expressions of statistical quenched pressure and self-consistency equations. 
\subsection{Generalities}
\label{HopfieldSection}
\begin{definition} 
\label{def:pspinham} 
Let $P\in\mathbb{N}$ and $\boldsymbol{\sigma}\in\lbrace-1;+1\rbrace^N$ be a configuration of N spins, the Hamiltonian of the P-spin glass model is defined as
	\begin{equation}
	H_N^{(P)}(\boldsymbol \sigma| \boldsymbol J) \coloneqq -\dfrac{1}{P!}\SOMMA{i_1,\cdots,i_P=1}{N,\cdots,N}\,J_{i_1,\cdots,i_P}\sigma_{i_1}\cdots\sigma_{i_P}
	\label{eq:SK_hbare}
	\end{equation}
where the P-wise quenched couplings $\boldsymbol{J} = \lbrace J_{i_1,\cdots,i_P}\rbrace_{i_1,\cdots,i_P=1,\cdots,N}$ are given by
\begin{equation}
    J_{i_1,...,i_P}\coloneqq\,\dfrac{J_0}{N^{P-1}}+J\sqrt{\dfrac{2}{N^{P-1}}}z_{i_1,\cdots,i_P}
\end{equation}
with $J_0 \in \mathbb{R}^+$ and $z_{i_1,\cdots,i_P}$ i.i.d. standard random variables $z_{i_1, \cdots , i_P} \sim N(0; 1)$.
\end{definition}

\begin{definition}
\label{SK_BareZ}
The partition function related to the Hamiltonian \eqref{eq:SK_hbare} is given by
\begin{equation}
    \begin{array}{lll}
    \label{eq:SK_BareZ}
	Z_N^{(P)}(\beta,\boldsymbol J) &\coloneqq \sum_{ \boldsymbol \sigma } \exp \left[ -\beta H_N^{(P)}(\boldsymbol \sigma | \boldsymbol J)\right]\,,
    \end{array}
\end{equation}
	where $\beta\in \mathbb R^+$ is the inverse temperature in proper units such that for $\beta \to 0$ the probability distribution for the neural configuration is uniformly spread while for $\beta \to \infty$ it is sharply peaked at the minima of the energy function \eqref{eq:SK_hbare}.
\end{definition}

We introduce the \emph{Boltzmann average} induced by the partition function (\ref{eq:SK_BareZ}), denoted with $\omega_{\boldsymbol J}$ and, for an arbitrary observable $O(\boldsymbol \sigma)$, defined as
	\begin{equation}
	\omega_{\boldsymbol J} (O (\boldsymbol \sigma)) : = \frac{\sum_{\boldsymbol \sigma} O(\boldsymbol \sigma) e^{- \beta H_N(\boldsymbol \sigma| \boldsymbol J)}}{Z_N(\beta, \boldsymbol J)}.
	\end{equation}
	This can be further averaged over the realization of the $J_{i_1,\cdots,i_P}$'s (also referred to as \emph{quenched average}) to get
	\begin{equation}
	  \langle O(\boldsymbol \sigma) \rangle \coloneqq \mathbb{E} \omega_{\boldsymbol J} (O(\boldsymbol \sigma)).  
	\end{equation}

\begin{definition}The intensive quenched pressure of the P-spin glass model (\ref{eq:SK_hbare}) is defined as
\begin{equation}
\label{PressureDef}
\mathcal A_N^{(P)}(\beta,\bm J) \coloneqq \frac{1}{N} \mathbb{E} \ln Z_N^{(P)}(\beta, \boldsymbol J),
\end{equation}
and its thermodynamic limit, assuming its existence, is referred to as
\begin{equation}
\mathcal A^{(P)}(\beta,\bm J) \coloneqq \lim_{N \to \infty} \mathcal A^{(P)}_N(\beta, \bm J).
\label{eq:quenched_pressure_def}
\end{equation}
\end{definition}

\begin{remark}
The existence of the thermodynamical limit of free energy density is proved for P-spin glass model in Guerra and Toninelli's work\cite{GuerraTon}.
\end{remark}

\begin{remark}
We stress that the quenched statistical pressure $\mathcal{A}^{(P)}$ is equivalent to free energy, namely 
\begin{align}
\mathcal{A}^{(P)} (\beta, \bm J) = -{\beta} f^{(P)}(\beta, \bm J) = -{\beta} \left(E(\beta, \bm J) - \frac{S(\beta, \bm J)}{\beta}\right)
\end{align}
where $E$ is the internal energy and S the entropy of the system. 
\end{remark}

In order to solve the model we want to find out an explicit expression for the quenched pressure \eqref{eq:quenched_pressure_def} in terms of the natural order parameters of the theory, namely the magnetization $m$ and the two-replica overlap $q_{12}$, defined in the following

\begin{definition} The order parameters used to describe the macroscopic behavior of the model are the standard ones \cite{Amit,Coolen}, namely, the magnetization 
	\begin{equation}
	m \coloneqq \frac{1}{N}\sum_{i=1}^{N} \sigma_i\,,
	\end{equation}
	and the two-replica overlap, introduced as
\begin{equation}
\label{q}
q_{12} \coloneqq \frac{1}{N}\sum_{i=1}^N \sigma_i^{(1)}\sigma_i^{(2)}\,.
\end{equation}
\end{definition}

\subsection{Resolution via Guerra's interpolation}
The purpose of this section is to solve the P-spin glass model through Guerra's interpolating technique. To do so, we compute the derivative w.r.t. the interpolating parameter $t$ and we apply the Fundamental theorem of Calculus. In the end, we find the expression of statistical pressure in the approximation of replica symmetry (RS) and first step of replica symmetry breaking (RSB).

\subsubsection{RS solution}
\begin{definition}
    \label{defn: RSassumption}
	Under the replica-symmetry assumption, the order parameters, in the thermodynamic limit, self-average around their mean values and their distributions get delta-peaked at their equilibrium value (denoted with a bar), independently of the replica considered, namely
	\begin{eqnarray}
	\label{eq:m_ter}
	\lim_{N\to \infty} \langle (m - \bar m)^2 \rangle = 0 &\Rightarrow& \lim_{N\to \infty}  \langle m \rangle = \bar m,\\
	\lim_{N\to \infty} \langle (q_{12} - \bar q)^2 \rangle = 0 &\Rightarrow& \lim_{N \to \infty}  \langle q_{12} \rangle = \bar q.
	\end{eqnarray}
	For the generic order parameter $X$ this can be rewritten as $\langle (\Delta X)^2 \rangle \xrightarrow[]{N\to\infty}0$,
	where
	$$
	\Delta X \coloneqq  X - \bar{X},
	$$
	and, clearly, the RS approximation also implies that, in the thermodynamic limit, $\langle \Delta X \Delta Y \rangle = 0$ for any generic pair of order parameters $X,Y$. Moreover in the thermodynamic limit, we have $\langle (\Delta X)^k \rangle \rightarrow 0$ for $k \geq 2$.
\end{definition}
\begin{definition} 
Given the interpolating parameter $t \in [0,1]$, $A,\ \psi \in \mathbb{R}$  and $z_i \sim \mathcal{N}(0,1)$ for $i=1, \hdots , N$ standard Gaussian variables i.i.d., the Guerra's interpolating partition function is given as 
\begin{equation}
\begin{array}{lll}
     \mathcal{Z}^{(P)}_N(t) &\coloneqq& \sommaSigma{\boldsymbol \sigma} \exp{}\Bigg[t\dfrac{\b J_0 N}{2} m^P(\boldsymbol \sigma)+(1-t)N\psi\,m(\boldsymbol \sigma)+
     \\\\
     & &+\sqrt{t}\b J\sqrt{\dfrac{1 }{2N^{P-1}}}\SOMMA{i_1,\cdots,i_{_{P}}=1}{N,\cdots,N}z_{i_1\cdots,i_{_{P}}}\sigma_{i_1}\cdots\sigma_{_{P}}+
     \\\\
     & &+A\sqrt{1-t}\,\SOMMA{i=1}{N} z_i\sigma_i\,,
     \label{def:partfunct_GuerraRS}
\end{array}
\end{equation}
where $\b=2\beta/P!$.
\end{definition}
\begin{definition} The interpolating pressure for the P-spin glass model (\ref{eq:SK_hbare}), at finite $N$, is introduced as
\begin{eqnarray}
\mathcal{A}^{(P)}_N(t) &\coloneqq& \frac{1}{N} \mathbb{E} \left[  \ln \mathcal{Z}^{(P)}_N(t)  \right],
\label{hop_GuerraAction}
\end{eqnarray}
where the expectation $\mathbb E$ is now meant over $z_{i_1,\cdots,i_P}$ and $z_i$, in the thermodynamic limit,
\begin{equation}
\mathcal{A}^{(P)}(t) \coloneqq \lim_{N \to \infty} \mathcal{A}^{(P)}_N(t).
\label{hop_GuerraAction_TDL}
\end{equation}
By setting $t=1$ the interpolating pressure recovers the original one (\ref{PressureDef}), that is $A_N^{(P)} (\beta,J_0,J) = \mathcal{A}^{(P)}_N(t=1)$.
\end{definition}

\begin{remark}
	The interpolating structure implies an interpolating measure, whose related Boltzmann factor reads as
	\begin{equation}
	\mathcal B (\boldsymbol \sigma; t )\coloneqq  \exp \left[ \beta \mathcal H (\boldsymbol \sigma; t) \right];
	\end{equation}

In this way the partition function is written as $\mathcal Z_N(t) =  \sum_{\boldsymbol \sigma} \mathcal B (\boldsymbol \sigma; t)$ .\\
A generalized average follows from this generalized measure as
\beq
	\omega_{t} (O (\boldsymbol \sigma )) \coloneqq   \sum_{\boldsymbol \sigma} O (\boldsymbol \sigma ) \mathcal B (\boldsymbol \sigma; t)
	\eeq
	and
\beq
\langle O (\boldsymbol \sigma ) \rangle_{t}  \coloneqq \mathbb E [ \omega_{t} (O (\boldsymbol \sigma)) ].
\eeq
where $ \mathbb E$ denotes the average over $z_{i_1,\cdots,i_P}$ and $\lbrace z_i\rbrace_{i=1,\cdots,N}$.

Of course, when $t=1$ the standard Boltzmann measure and related average is recovered.
Hereafter, in order to lighten the notation, we will drop the subindices $t$.
\end{remark}
\begin{lemma} 
The $t$ derivative of interpolating pressure is given by 
\begin{equation}
\small
    \begin{array}{lll}
         \dfrac{d \mathcal{A}^{(P)}(t)}{d t}\coloneqq & \dfrac{\b J_0 }{2} \left(\l m^P\r-\dfrac{2\psi }{\b J_0}\l m \r\right)-\dfrac{\beta'\,^2 J^2}{4} \left(\l q_{12}^P\r-\dfrac{2 A^2 }{\beta'\,^2 J^2 }\l q_{12}\r\right)
         \\\\
         &+\dfrac{\beta'\,^2 J^2}{4}-\dfrac{A^2}{2}
    \end{array}
    \label{eq:streaming_RS_Guerra}
\end{equation}
\end{lemma}

\begin{proof}
Deriving equation \eqref{hop_GuerraAction} with respect to $t$, we get 
\begin{equation}
\footnotesize
\begin{array}{lll}
     \dfrac{d \mathcal{A}^{(P)}(t)}{d t}&=&\dfrac{\b J_0}{2}\l m^P \r - \psi \l m \r+
     \\\\
     &&+\dfrac{\b J}{2N \sqrt{t}}\sqrt{\dfrac{1}{ 2N^{P-1}}}\mathbb{E}\left[\sum\limits_{\boldsymbol{i}}z_{\boldsymbol{i}}\omega(\sigma_{i_1}\cdots \sigma_{i_{P}})\right]+
     \\\\
     &&-\dfrac{1}{2N\sqrt{1-t}}\mathbb{E}\left[A\sum\limits_{i} z_{i}\omega(\sigma_{i})\right]\,.
\end{array}
\label{eq:proof_streaming_Guerra}
\end{equation}
Now, using the Stein's lemma (also known as Wick's theorem) on standard Gaussian variable $z_i$ and $z_{\boldsymbol{i}}$, which is for a standard Gaussian variable $J$, i.e. $J \sim N(0, 1)$, and for a generic
function $f(J)$ for which $\mathbb{E} z f(z)$ and $\mathbb{E} \partial_J f(J)$ both exist, then
\begin{align}
    \label{eqn:gaussianrelation2}
    \mathbb{E} \left( J f(J)\right)= \mathbb{E} \left( \frac{\partial f(J)}{\partial J}\right)\,,
\end{align}
we may rewrite the second and the third member of \eqref{eq:proof_streaming_Guerra} as
\begin{equation*}
\footnotesize
\begin{array}{lll}
     \dfrac{\b J}{2N \sqrt{t}}\sqrt{\dfrac{1}{ 2N^{P-1}}}\mathbb{E}\left[\sum\limits_{\boldsymbol{i}}z_{\boldsymbol{i}}\omega(\sigma_{i_1}\cdots \sigma_{i_{P}})\right]-\dfrac{1}{2N\sqrt{1-t}}\mathbb{E}\left[A\sum\limits_{i} z_{i}\omega(\sigma_{i})\right]=D_1+D_2\,.
\end{array}
\end{equation*}
Let’s investigate those three terms:
\begin{equation}
\footnotesize
\begin{array}{lll}
     D_1&=\dfrac{\b J}{2N \sqrt{t}}\sqrt{\dfrac{1}{ 2N^{P-1}}}\mathbb{E}\left[\sum\limits_{\boldsymbol{i}}\partial_{z_{\boldsymbol{i}}}\omega(\sigma_{i_1}\cdots \sigma_{i_{P}})\right]
     \\\\
     &=\dfrac{\beta'\,^2 J^2}{4 N^{P}}\left(\sum\limits_{\boldsymbol{i}}\mathbb{E}\left[\omega((\sigma_{i_1}\cdots \sigma_{i_{P}})^2)\right]-\sum\limits_{\boldsymbol{i}}\mathbb{E}\left[\omega(\sigma_{i_1}\cdots \sigma_{i_{P}})^2\right]\right)
     \\\\
     &=\dfrac{\beta'\,^2 J^2}{4 }\Big[1-\l\qq^{P}\r\Big]\,;
\end{array}
\label{eq:D1_GuerraRS}
\end{equation}
\begin{equation}
\footnotesize
\begin{array}{lll}
     D_2&=
    -\dfrac{1}{2N\sqrt{1-t}}A\sum\limits_{i}\mathbb{E}\left[ \partial_{z_{i}}\omega(\sigma_{i})\right]
    \\\\
    &=-\dfrac{1}{2N}A^2\left(\sum\limits_{i}\mathbb{E}\left[ \omega(\sigma^{2}_{i})\right]-\sum\limits_{i}\mathbb{E}\left[ \omega(\sigma_{i})^{2}\right]\right)
     \\\\
     &=-\dfrac{1}{2}A^2\Big[1-\l\qq\r\Big]\,.
\end{array}
\label{eq:D2_GuerraRS}
\end{equation}
Rearranging together \eqref{eq:D1_GuerraRS} and \eqref{eq:D2_GuerraRS} we obtain the thesis.
\end{proof}

\begin{remark}
We stress that, for the RS assumption presented in Definition \eqref{defn: RSassumption}, we can use the relations
\footnotesize
\begin{align}
    \langle m^P \rangle &= \sum_{k=2}^P \begin{pmatrix}P\\k\end{pmatrix} \langle (m-\m)^k \rangle \m^{P-k} + \m^P (1-P) + P\m^{P-1}\langle m \rangle\,, \label{potential_m1}
    \\
   \langle q_{12}^P \rangle &= \sum_{k=2}^P \begin{pmatrix}P\\k\end{pmatrix} \langle (q_{12}-\q)^k \rangle \q^{P-k} + \q^P (1-P) + P\q^{P-1}\langle q_{12} \rangle\,,
    \label{potential_pq}
\end{align}
\normalsize
which are proven in Appendix \eqref{app:potenziali}. Using these relations, if we fix the constants as
\begin{equation}
    \begin{array}{lll}
         \psi=\dfrac{P}{2}\b J_0 \m^{^{P-1}}\,,
         \\\\
         A^2=\dfrac{P}{2}\beta'\,^2 J^2\q^{^{P-1}}\,,
    \end{array}
    \label{eq:constant_GuerraRS}
\end{equation}
the \eqref{eq:streaming_RS_Guerra} in the thermodynamical limit reads as
\begin{equation}
\begin{array}{lll}
     \dfrac{d \mathcal{A}^{(P)}(t)}{d t}&\coloneqq &-\dfrac{P-1}{2}\b J_0 \m^P +\dfrac{\beta'\,^2J^2}{4}\left(1-P\q^{P-1}+(P-1)\q^P\right)
\end{array}
\label{eq:streaming_RS_Guerra2}
\end{equation}
which is now independent of $t$.
\end{remark}

\begin{proposition}
\label{SK_P_quenched}
In the thermodynamic limit ($N\to\infty$) and under RS assumption, applying the Fundamental Theorem of Calculus and using the suitable values of $A$ and $\psi$, we find the quenched pressure for the P-spin glass model as
\begin{equation}
\label{eq:pressure_GuerraRS}
\begin{array}{lll}
     \mathcal{A}^{(P)}(\beta,J_0,J) &=& \ln{2} +\left\langle\ln{\cosh{\left[\dfrac{P}{2}\b J_0 \m^{^{P-1}} +z \b J\sqrt{\dfrac{P}{2}\q^{^{P-1}}}\right]}}\right\rangle_z+
    \\\\
    &&-\dfrac{P-1}{2 }\b J_0 \m^P +\dfrac{\beta'\,^2J^2}{4}\left(1-P\q^{P-1}+(P-1)\q^P\right)\,,
\end{array}
\end{equation}
where the values of order parameters are ruled by the following self-consistence equations 
\begin{equation}
    \begin{array}{lll}
         \m=\left\langle\tanh{\left[\dfrac{P}{2}\beta J_0 \m^{^{P-1}} +z \beta J\sqrt{\dfrac{P }{2}\q^{^{P-1}}}\right]}\right\rangle_z \,,
         \\\\
         \q=\left\langle\tanh{}^2{\left[\dfrac{P}{2}\beta J_0 \m^{^{P-1}} +z \beta J\sqrt{\dfrac{P }{2}\q^{^{P-1}}}\right]}\right\rangle_z \,.
    \end{array}
    \label{eq:self_GuerraRS}
\end{equation}
\end{proposition}

\begin{proof}
Using the Fundamental Theorem of Calculus:
\begin{equation}
\footnotesize
    \mathcal{A}^{(P)}=\mathcal{A}^{(P)}(t=1)=\mathcal{A}^{(P)}(t=0)+\int\limits_0^1\, \partial_s \mathcal{A}^{(P)}(s
)\Big|_{s=t}\,dt
\label{eq:F_T_Calculus}
\end{equation}
and computing the one-body terms
\begin{equation}
\footnotesize
\begin{array}{lll}
     \mathcal{A}^{(P)}(t=0)&=&\dfrac{1}{N}\mathbb{E}\ln{}\Bigg\lbrace \sommaSigma{\boldsymbol \sigma} \exp{}\Bigg[N\psi\,m(\boldsymbol \sigma)+A\,\SOMMA{i=1}{N} z_i\sigma_i\Bigg]\Bigg\rbrace\,,
     \\\\
         &=&\ln{2} +\left\langle\ln{\cosh{\left[\psi +z A\right]}}\right\rangle_z+
       \\\\
        &=&\ln{2} +\left\langle\ln{\cosh{\left[\dfrac{P}{2}\b J_0 \m^{^{P-1}} +z \b J\sqrt{\dfrac{P}{2}\q^{^{P-1}}}\right]}}\right\rangle_z\,.
    \end{array}
    \label{eq:one_body_GuerraRS}
\end{equation}
Finally, putting \eqref{eq:streaming_RS_Guerra2} and \eqref{eq:one_body_GuerraRS} in \eqref{eq:F_T_Calculus}, we find \eqref{eq:pressure_GuerraRS}.

Extremizing the statistical pressure in \eqref{eq:pressure_GuerraRS} w.r.t. the order parameters we find the self-consistency equations.
\end{proof}

\begin{remark}
If we do not take into account the disordered part of the system, namely $J=0$, in \eqref{eq:pressure_GuerraRS}, we recover for this P-spin ferromagnetic model the same result in \cite{fachechi}.
\end{remark}

\begin{remark}
If we don't consider  in \eqref{eq:pressure_GuerraRS} the ferromagnetic part of the system, namely $J_0=0$, we recover the same expressions found in \cite{Adriano_HJ}.
\end{remark}

\subsubsection{1RSB solution}

In this subsection we turn to the solution of the P-spin glass model via the Guerra's interpolating technique, restricting the description at the first step of RSB.

\begin{definition} \label{def:HM_RSB}
In the first step of replica-symmetry breaking, the distribution of the two-replica overlap $q$, in the thermodynamic limit, displays two delta-peaks at the equilibrium values, referred to as $\bar{q}_1,\ \bar{q}_2$, and the concentration on the two values is ruled by $\theta \in [0,1]$, namely
\begin{equation}
\lim_{N \rightarrow + \infty} P'_N(q) = \theta \delta (q - \bar{q}_1) + (1-\theta) \delta (q - \bar{q}_2). \label{limforq2}
\end{equation}
The magnetization still self-averages at $ \bar{m}$ as in (\ref{eq:m_ter}).
\end{definition}

Following the same route pursued in the previous sections, we need an interpolating partition function $\mathcal{Z}$ and an interpolating quenched pressure $\mathcal{A}$,  that are defined hereafter.
\begin{definition}
Given the interpolating parameter $t$ and the i.i.d. auxiliary fields $\lbrace z_i^{(1)}, z_i^{(2)}\rbrace_{i=1,...,N}$, with $z_i^{(1,2)} \sim \mathcal N(0,1)$ for $i=1, ..., N$ we can write the 1-RSB interpolating partition function $\mathcal Z_N(t)$ for the P-spin glass model (\ref{eq:SK_hbare}) recursively, starting by

\begin{equation}
\begin{array}{lll}
     \mathcal{Z}^{(P)}_2(t) &:=& \sommaSigma{\boldsymbol \sigma} \displaystyle\int \mathcal{D} \bm \tau\exp{}\Bigg[t\dfrac{\b J_0 N}{2} m^P(\boldsymbol \sigma)+(1-t)N\psi\,m(\boldsymbol \sigma)+
     \\\\
     & &+\sqrt{t}\b J\sqrt{\dfrac{1 }{2N^{P-1}}}\SOMMA{i_1,\cdots,i_{_{P}}=1}{N,\cdots,N}z_{i_1\cdots,i_{_{P}}}\sigma_{i_1}\cdots\sigma_{i_{P}}+
     \\\\
     & &+\sqrt{1-t}\SOMMA{a=1}{2}\left(A_a\SOMMA{i=1}{N} z_i^{(a)}\sigma_i\right)
     \label{def:partfunct_GuerraRSB}
\end{array}
\end{equation}
where  the $z_{i_1\cdots,i_{_{P}}}$'s i.i.d. standard Gaussian. 
Averaging out the fields recursively, we define
\begin{align}
\label{eqn:Z1}
\mathcal Z_1^{(P)}(t) \coloneqq& \mathbb E_2 \left [ \mathcal Z_2^{(P)}(t)^\theta \right ]^{1/\theta} \\
\label{eqn:Z0_1RSB}
\mathcal Z_0^{(P)}(t) \coloneqq&  \exp \mathbb E_1 \left[ \ln \mathcal Z_1^{(P)}(t) \right ] \\
\mathcal Z_N^{(P)}(t) \coloneqq & \mathcal Z_0^{(P)}(t) ,
\end{align}
where with $\mathbb E_a$ we mean the average over the variables $z_i^{(a)}$'s, for $a=1, 2$, and with $\mathbb{E}_0$ we shall denote the average over the variables $z_{i_1\cdots,i_{_{P}}}$'s.
\end{definition}

\begin{definition}
\label{def:interpPressRSB}
The 1RSB interpolating pressure, at finite volume $N$, is introduced as
\begin{equation}\label{AdiSK1RSB}
\mathcal A_N^{(P)} (t) \coloneqq \frac{1}{N}\mathbb E_0 \big[ \ln \mathcal Z_0^{(P)}(t) \big],
\end{equation}
and, in the thermodynamic limit, assuming its existence
\begin{equation}
\mathcal A^{(P)} (t) \coloneqq \lim_{N \to \infty} \mathcal A^{(P)}_N (t).
\end{equation}
By setting $t=1$, the interpolating pressure recovers the standard pressure (\ref{PressureDef}), that is, $A^{(P)}_N(\beta,J_0,J) = \mathcal A^{(P)}_N (t =1)$.
\end{definition}

\begin{remark}
\label{rem:medie}
In order to lighten the notation, hereafter we use the following 
\begin{align}
\label{eq:unouno_a}
\langle m \rangle=& \mathbb E_0  \mathbb E_1  \mathbb E_2  \left[\mathcal W_2\frac{1}{N}\sum_{i=1}^N \omega( \sigma_i) \right] \\
\label{eqn:q121_a}
 \langle q_{12} \rangle_1=&\mathbb E_0  \mathbb E_1  \left[\frac{1}{N} \sum_{i=1}^N \left( \mathbb E_2 \left[\mathcal W_2\omega(\sigma_i)\right] \right)^2 \right] \\
\label{eqn:q122_a}
\langle q_{12} \rangle_2=& \mathbb E_0  \mathbb E_1  \mathbb E_2 \left [ \mathcal W_2\frac{1}{N}\sum_{i=1}^N \omega(\sigma_i)^2 \right] 
\end{align}
where the weight $\mathcal W_2$ is defined as
\begin{equation}
\mathcal W_2 :=\frac{\mathcal Z_2^\theta}{\mathbb E_2 \left [\mathcal Z_2^\theta\right]}.
\end{equation}
\end{remark}

Now the next step is computing the $t$-derivative of the statistical pressure. In this way we can apply the fundamental theorem of calculus and find the solution of the original model. 

\begin{lemma}
\label{lemma:tderRSB}
The derivative w.r.t. $t$ of interpolating statistical pressure can be written as 
\begin{equation}
\small
\begin{array}{lll}
      d_t \mathcal{A}^{(P)}_N =& \dfrac{\b J_0}{2} \langle m^P \rangle -\psi \langle m \rangle + \dfrac{\beta'\,^2 J^2 }{4} \left[ 1+ (\theta-1) \langle q_{12}^P \rangle_2 - \theta \langle q_{12}^P \rangle_1 \right]  
 \\\\
    &- (A_1)^2 \left[ 1+(\theta-1) \langle q_{12} \rangle_2 - \theta \langle q_{12} \rangle_1 \right]- (A_2)^2 \left[ 1+(\theta-1) \langle q_{12} \rangle_2 \right]    
\end{array}
\end{equation}
\end{lemma}

Since the proof is a bit lenghty, we leave it in Appendix \eqref{app:tder1RSB}. 

\begin{remark}
If we want to use the 1RSB assumpion in such a way that
for $a=1,2$
\footnotesize
\begin{align}
    \langle m^P \rangle &= \sum_{k=2}^P \begin{pmatrix}P\\k\end{pmatrix} \langle (m-\bar{m})^k \rangle \bar{m}^{P-k} + \bar{m}^P (1-P) + P\bar{m}^{P-1}\langle m \rangle\,, \label{potential_m_1RSB}
    \\
    \langle q_{12}^P \rangle_a &= \sum_{k=2}^P \begin{pmatrix}P\\k\end{pmatrix} \langle (q_{12}-\q_a)^k \rangle _a\q_a^{P-k} + \q_a^P (1-P) + P\q_a^{P-1}\langle q_{12} \rangle_a\,,
    \label{potential_pq_1RSB}
\end{align}
\normalsize
we need to fix the constants as 
\begin{align}
\label{value_psi}
    \psi&=  \frac{\b J_0}{2}P \m^{P-1},\\
    A_1 &= \frac{\beta'\,^2 J^2 P }{4} \q_1^{P-1},\\
    A_2 &= \frac{\beta'\,^2 J^2 P}{4} (\q_2^{P-1}-\q_1^{P-1}).
\end{align}
Thus, the derivative w.r.t. $t$ in the thermodynamic limit is computed as 
\begin{equation}
    \label{dtA_1RSB_thermo}
    \begin{array}{lll}
         d_t \mathcal{A}^{(P)}_N =&  \dfrac{\b J_0}{2}\m^P(1-P) + \dfrac{\beta'\,^2 J^2}{4} \left[1-P\q_2^{P-1}+(P-1)\q_2^P \right]
         \\\\
    &-\dfrac{\beta'\,^2 J^2 }{4}(P-1)\theta(\q_2^P-\q_1^P)\,.
    \end{array}
\end{equation}
\end{remark}

\begin{proposition}
At finite size and under 1RSB assumption applying the Fundamental Theorem of Calculus and using the suitable values of $A^{(1)}, A^{(2)}, \psi$, we find the quenched pressure for the P-spin glass model as
\begin{equation}
    \begin{array}{lll}
\label{A_1RSB_finite}
    &\mathcal{A}^{(P)} = \log 2 +\dfrac{\b J_0}{2}\m^P(1-P) + \dfrac{\beta'\,^2 J^2}{4} \left[1-P\q_2^{P-1}+(P-1)\q_2^P \right]
    \\\\
    &+ \dfrac{1}{\theta} \mathbb{E}_1 \log \mathbb{E}_2 \cosh^\theta \left[\b\left( \dfrac{ J_0}{2} P\m^{P-1} + z^{(1)}  J \sqrt{\frac{P}{2} \q_1^{P-1}}+ z^{(2)}  J \sqrt{\frac{P}{2} (\q_2^{P-1}-\q_1^{P-1}})\right) \right] 
    \\\\
    &-\dfrac{\beta'\,^2 J^2 }{4}(P-1)\theta(\q_2^P-\q_1^P) + \dfrac{\b J_0 P}{2} \SOMMA{k=2}{P} \begin{pmatrix}P\\k\end{pmatrix} \langle (\Delta m)^k \rangle \bar{m}^{P-k}  
    \\\\
    &+ \dfrac{\beta'\,^2 J^2}{4}(\theta-1)\SOMMA{k=2}{P} \begin{pmatrix}P\\k\end{pmatrix} \langle (\Delta q_2)^k \rangle _2\q_2^{P-k}-\dfrac{\beta'\,^2 J^2}{4}\theta \SOMMA{k=2}{P} \begin{pmatrix}P\\k\end{pmatrix} \langle (\Delta q_1)^k \rangle _1\q_1^{P-k}
    \end{array}
\end{equation}
where we use $\Delta X_a=X-\bar{X}_a$.
\end{proposition}

\begin{proof}
We use the fundamental theorem of Calculus
\begin{align}
    \mathcal{A}^{(P)} = \mathcal{A}^{(P)}(t=1) = \mathcal{A}^{(P)}(t=0)+\int_0^1 \partial_s \mathcal{A}^{(P)}(s) \vert_{s=t} dt.
\end{align}
The last step we need is the computation of one-body term. We omit it since it is similar to RS case.  
\end{proof}

\begin{theorem}
In the thermodynamic limit, under 1RSB assumption, the expression of quenched statistical pressure is

\begin{equation}
    \begin{array}{lll}
          \mathcal{A}^{(P)}=& \log 2 -\dfrac{\b J_0}{2}\m^P(P-1) + \dfrac{\beta'\,^2 J^2}{4} \left[1-P\q_2^{P-1}+(P-1)\q_2^P \right]
          \\\\
    &-\dfrac{\beta'\,^2 J^2 }{4}(P-1)\theta(\q_2^P-\q_1^P)+ \dfrac{1}{\theta} \mathbb{E}_1 \log \mathbb{E}_2 \cosh^\theta g(\bm z, \bar{m}) \,,
    \label{A_1RSB_finalissima}
    \end{array}
\end{equation}

where 
\begin{align}
\small
    g(\bm z, \bar{m})&= \frac{\b J_0}{2} \m^{P-1} + z^{(1)} \b J \sqrt{ \frac{P}{2} \q_1^{P-1}}+ z^{(2)} \b J \sqrt{ \frac{P}{2} (\q_2^{P-1}-\q_1^{P-1}})
    \label{def_g}
\end{align}

where the order parameters are ruled by
\begin{align}
&\bar{m}= \mathbb{E}_1 \left[ \frac{\mathbb{E}_2 \cosh^\theta g(\bm z,\bar{m})\tanh g(\bm z,\bar{m})}{\mathbb{E}_2 \cosh^\theta g(\bm z,\bar{m})}\right] \\
&\bar{q}_1 = \mathbb{E}_1 \left[ \frac{\mathbb{E}_2 \cosh^\theta g(\bm z,\bar{m})\tanh g(\bm z,\bar{m})}{\mathbb{E}_2 \cosh^\theta g(\bm z,\bar{m})}\right]^2 \\
&\bar{q}_2 = \mathbb{E}_1 \left[ \frac{\mathbb{E}_2 \cosh^\theta g(\bm z,\bar{m})\tanh^2 g(\bm z,\bar{m})}{\mathbb{E}_2 \cosh^\theta g(\bm z,\bar{m})}\right]
\end{align}
where $g(\bm z, \m)$ is the same defined in \eqref{def_g}.
\end{theorem}

\begin{proof}

If we use definition \ref{defn: RSassumption}, the last three terms of expression \eqref{A_1RSB_finite} will vanish in the thermodynamic limit, so we recover \eqref{A_1RSB_finalissima}.

Extremizing the statistical pressure in \eqref{A_1RSB_finalissima} w.r.t. the order parameters we find the self-consistency equations.
\end{proof}

\begin{remark}
The expression \eqref{A_1RSB_finalissima} with $J_0=0$ is the same computed by \cite{gardner, GrossMezard} through replica trick and by \cite{Adriano_HJ} via Hamilton-Jacobi technique in 1RSB assumption. 
\end{remark}

\subsection{Resolution via transport equation}
\label{sec:transport1}
Now the aim is to show that, using a different strategy, such as the transport equation technique, already introduced in \cite{lindaRSB} for both assumption in classic SK model with ferromagnetic contribution, we are able to retrieve the same expression of quenched statistical pressure and self-consistency equations found via Guerra's interpolating technique. 

\subsubsection{RS solution} \label{ssec:HRS}
In this subsection we find the RS solution of P-spin glass model via the transport equation technique. The strategy is to introduce an interpolating pressure $\mathcal A_N$ living in a fictitious space-time framework and recovering the intensive quenched pressure $\mathcal A_N$ of the original model in a specific point of this space, and to show that it fulfills a transport equation in such a way that the solution of the statistical mechanical problem is recast in the solution of a partial differential equation.

The definition of RS assumption for the order parameters is the same as \eqref{defn: RSassumption}.
\begin{definition} 
Given the interpolating parameter $t, x, w$, and $z_i \sim \mathcal{N}(0,1)$  standard Gaussian variables iid, the partition function is given as 
\begin{equation}
\footnotesize
\begin{array}{lll}
     \mathcal{Z}^{(P)}_N(t, \bm r) &\coloneqq& \sommaSigma{\boldsymbol \sigma} \exp{}\Bigg[t\dfrac{\b J_0 N}{2} m^P(\boldsymbol \sigma)+w N\,m(\boldsymbol \sigma)+
     \\\\
     & &+\sqrt{t}\b J\sqrt{\dfrac{1 }{2N^{P-1}}}\SOMMA{i_1,\cdots,i_{_{P}}=1}{N,\cdots,N}z_{i_1\cdots,i_{_{P}}}\sigma_{i_1}\cdots\sigma_{_{P}}+
     \\\\
     & &+\sqrt{x}\,\SOMMA{i=1}{N} z_i\sigma_i\,,
     \label{def:partfunct_transpRS}
\end{array}
\end{equation}
where we set $\beta'=2 \beta/P!$.
\end{definition}
Similar to RS Guerra's interpolation, we can define the interpolating pressure, the Boltzmann factor and the generalized measure. 

\begin{lemma} 
The partial derivatives of the interpolating pressure (\ref{hop_GuerraAction}) w.r.t. $t,x,w$ give the following expectation values:
	\bea
	\label{hop_expvalsa}
	\frac{\partial \mathcal{A}^{(P)}_N}{\partial t} &=& \dfrac{\b J_0}{2}\l m^P \r +\dfrac{\beta'\,^2J^2}{4}\Big(1-\l\qq^{P}\r\Big),
	\\
	\label{hop_expvalsmiddle}
	\frac{\partial \mathcal{A}^{(P)}_N}{\partial x}  &=& \dfrac{1}{2}\Big(1-\l\qq\r\Big),\\
	\frac{\partial \mathcal{A}^{(P)}_N}{\partial w}  &=&  \l m\r.
	\label{hop_expvalsb}
	\eea
\end{lemma}

\begin{proof}
Since the procedures for the derivatives w.r.t. each parameter are analogous, we prove only the derivatives w.r.t. $t$. The  partial  derivative  of  the  interpolating  quenched  pressure  with respect to $t$ reads as 
\begin{equation}
    \footnotesize
    \begin{array}{lll}
           \dfrac{\partial \mathcal{A}^{(P)}_N}{\partial t}&=&\dfrac{1}{N}\mathbb{E}\left[\dfrac{\beta'N}{2}\omega(m^P)\right]+
           \\\\
           & & + \dfrac{\b J}{2N\sqrt{t}}\sqrt{\dfrac{1}{2N^{P-1}}}\sum\limits_{\boldsymbol{i}}\mathbb{E}\left[z_{\boldsymbol{i}}\omega(\sigma_{i_1}\cdots\sigma_{i_{P}})\right]\,.
    \end{array}
\end{equation}
Now, using the Stein's lemma \eqref{eqn:gaussianrelation2} on standard Gaussian variable $z_{\boldsymbol i}$
\begin{equation}
    \footnotesize
    \begin{array}{lll}
           \dfrac{\partial \mathcal{A}^{(P)}_N}{\partial t}&=&\dfrac{\b J_0}{2}\l m^P\r+ \dfrac{\beta'\,^2J^2}{4N^{^{P}}}\sum\limits_{\boldsymbol i}\left(\mathbb{E}\left[\omega((\sigma_{i_1}\cdots\sigma_{i_{P}})^2)\right]-\mathbb{E}\left[\omega(\sigma_{i_1}\cdots\sigma_{i_{P}})^2\right]\right)
           \\\\
           &=&\dfrac{\b J_0}{2}\l m^P\r+ \dfrac{\beta'\,^2 J^2}{4}\left(1-\l\qq^{P}\r\right)\,.
    \end{array}
\end{equation}
\end{proof}
%
%
\begin{proposition} \label{prop:interp_transp_RS}
The interpolating pressure \eqref{def:partfunct_transpRS} at finite size obeys the following differential equation:
	\begin{equation}
	\frac{d \mathcal{A}^{(P)}_N}{dt} =  \pder{\mathcal{A}_N^{(P)}}{t} + \dot{x} \pder{\mathcal{A}^{(P)}_N}{x} +\dot w \pder{\mathcal{A}^{(P)}_N}{w}= S(t, \bm r)+V_N(t, \bm r),
	\label{hop_GuerraAction_DE}
	\end{equation}
	where we set 
	\begin{equation}
	    \begin{array}{lll}
	         \dot x = -\dfrac{P}{2}\beta'\,^2J^2\q^{P-1}\,, && \dot w = -\dfrac{P}{2}\b J_0\m^{P-1}
	    \end{array}
	    \label{eq:deriv_r_trasp_RS}
	\end{equation}
	and
	\begin{equation}
	\begin{array}{lll}
	     S(t, \bm r) &\coloneqq& -\dfrac{P-1}{2}\b J_0 \m^P+\dfrac{\beta'\,^2J^2}{4}\left(1-P\q^{P-1}+(P-1)\q^P\right),
	\end{array}
	\end{equation}
	\begin{equation}
	\begin{array}{lll}
	     V_N(t, \bm r) &\coloneqq& \dfrac{\b J_0}{2}\SOMMA{k=2}{P}\begin{pmatrix}
	    P\\k
	\end{pmatrix} \m^{P-k}\l(\Delta m)^k\r +\dfrac{\beta'\,^2 J^2}{4}\SOMMA{k=2}{P}\begin{pmatrix}
	    P\\k
	\end{pmatrix} \q^{P-k}\l(\Delta q_{12})^k\r.
	\end{array}
	\label{potenziale-RS-Hopfield}
	\end{equation}
	\end{proposition}
	
	\begin{proof}
	Starting to evaluate explicitly $\dt \mathcal{A}_N$ by using (  \eqref{hop_expvalsmiddle} - \eqref{hop_expvalsb} and Definition \eqref{potential_m1}-\eqref{potential_pq}) we write
\begin{equation}
\footnotesize
    \begin{array}{lll}
         \dfrac{\partial}{\partial t}\mathcal{A}^{(P)}_N=\dfrac{\b J_0}{2}\left(\SOMMA{k=2}{P} \begin{pmatrix}P\\k\end{pmatrix} \langle (m-\bar{m})^k \rangle \bar{m}^{P-k} + \bar{m}^P (1-P) + P\bar{m}^{P-1}\langle m\rangle\right) 
         \\\\
         +\dfrac{\beta'\,^2 J^2}{4}\left(1-\SOMMA{k=2}{P} \begin{pmatrix}P\\k\end{pmatrix} \langle (q_{12}-\q)^k \rangle \q^{P-k} + \q^P (1-P) + P\q^{P-1}\langle q_{12} \rangle\right) 
         \\\\
         =V_N(t,\boldsymbol{r})+S(t,\boldsymbol{r})+\dfrac{P}{2}\b J_0\bar{m}^{P-1}\l m \r +\dfrac{\beta'\,^2 J^2}{2}P\q^{P-1}\dfrac{1}{2}\left(1 - \langle q_{12} \rangle\right) 
         \\\\
         =V_N(t,\boldsymbol{r})+S(t,\boldsymbol{r}) +\dfrac{\b J_0}{2}P\bar{m}^{P-1}\Big(\frac{\partial \mathcal{A}^{(P)}_N}{\partial w} \Big)+\dfrac{\beta'\,^2 J^2}{2}P\q^{P-1}\Big(\frac{\partial \mathcal{A}^{(P)}_N}{\partial x} \Big)  
    \end{array}
\end{equation}
Thus, by placing $\dot{\boldsymbol{r}}=(\dot{x},\dot{w})$ as in \eqref{eq:deriv_r_trasp_RS} we reach the thesis.

	\end{proof}
	
	\begin{remark}
	\label{remark:Vthermolim}
	Since the RS assumption, which is linked to the fact that in the thermodynamic limit
	\begin{align*}
	    \langle (m-\bar{m})^k \rangle &= 0 \: \ \: k \geq 2  \\
	    \langle (q_{12} - \bar{q})^k \rangle &= 0 \: \ \: k \geq 2
	\end{align*}
	we can say that $V_N$ defined in \eqref{potenziale-RS-Hopfield} vanishes. 
	\end{remark}
	
	\begin{proposition}
The transport equation associated to the interpolating pressure $\mathcal{A}_N(t, \bm r)$ in the thermodynamic limit and under the RS assumption is
	\begin{equation}
	\footnotesize
	\begin{array}{lll}
	     \pder{\mathcal{A}^{(P)}_{\textrm{RS}}}{t}-\dfrac{\beta'\,^2 J^2}{2}P\q^{P-1} \left(\pder{\mathcal{A}^{(P)}_{\textrm{RS}}}{x}\right)-\dfrac{P}{2}\b J_0\m^{P-1} \left(\pder{\mathcal{A}^{(P)}_{\textrm{RS}}}{w}\right)= \\\ -\dfrac{P-1}{2}\b J_0 \m^P+\dfrac{\beta'\,^2J^2}{4}\left(1-P\q^{P-1}+(P-1)\q^P\right),
	\end{array}
	\label{hop_GuerraAction_RSDE}
	\end{equation}
	whose solution is given by
\begin{equation}
\small
\begin{array}{lll}
     \mathcal A^{(P)}_{\rm RS}(t,\boldsymbol{r})=&\ln{2}+\left\l\ln{\cosh{\left[w+\dfrac{P}{2}\b J_0\m^{P-1}t+z\sqrt{x+\dfrac{P}{2}\beta'\,^2J^2\q^{P-1}t}\right]}}\right\r_{z}
     \\\\
     \textcolor{white}{A^{(P)}_{\rm RS}(t,\boldsymbol{r})=}&-\dfrac{P-1}{2}\b J_0 \m^Pt+\dfrac{\beta'\,^2J^2}{4}t\left(1-P\q^{P-1}+(P-1)\q^P\right)\,.
\end{array}
\label{hop_mechanicalsolution}
\end{equation}

\end{proposition}	

\begin{proof}
We can find the transport equation applying Remark \eqref{remark:Vthermolim} in Prop. \eqref{prop:interp_transp_RS}. 

We compute the solution using the characteristic method on the transport equation: 
\begin{align}
    \mathcal{A}^{(P)}_{RS}(t, \bm r) = \mathcal{A}^{(P)}_{RS}(0, \bm r-\dot{\boldsymbol{r}}t) + S(t, \bm r)t.
\end{align}
where $\dot{\boldsymbol{r}}=(\dot{x},\dot{w})$. Along the characteristics, the fictitious motion in the $(t,\boldsymbol{r})$ time-space is linear and returns
\begin{equation}
    \begin{array}{lll}
         x=x_0-\dfrac{P}{2}\beta'\,^2J^2\q^{P-1}t &&  w=w_0-\dfrac{P}{2}\b J_0\m^{P-1}t
    \end{array}
\end{equation}
where $\boldsymbol{r}_0=(x_0,w_0)=(x(t=0),w(t=0))$. The Cauchy condition at $t=0$ is given by a direct computation at finite $N$ as
\begin{equation}
\begin{array}{lll}
     \mathcal A^{(P)}(0,\boldsymbol{r}-\dot{\boldsymbol{r}}t)&=&\mathcal A^{(P)}(0,\boldsymbol{r}_0)
     \\\\
     & =&\dfrac{1}{N}\mathbb{E}\Bigg\lbrace\sommaSigma{\boldsymbol \sigma} \exp{}\Bigg[w_0 N\psi\,m(\boldsymbol \sigma)+\sqrt{x_0}\SOMMA{i=1}{N} z_i\sigma_i\Bigg]\Bigg\rbrace
     \\\\
    & =&\ln{2}+\left\l\ln{\cosh{\left[\omega_0+z\sqrt{y_0}\right]}}\right\r_{z}\,.
\end{array}
\end{equation}
\end{proof}

Giving the suitable values of parameters we have the following
 \begin{corollary}
 \label{cor:A_RS_trans}
 The RS approximation of the quenched pressure for the P-spin glass model is obtained by posing $t=1$ and $\bm r = \bm 0$ in (\ref{hop_mechanicalsolution}), which returns
 \begin{equation}
\begin{array}{lll}
     \mathcal A^{(P)}_{\rm RS}(t,\boldsymbol{r})=&\ln{2}+\left\l\ln{\cosh{\left[\dfrac{P}{2}\b J_0\m^{P-1}+z\sqrt{\dfrac{P}{2}\beta'\,^2J^2\q^{P-1}}\right]}}\right\r_{z}
     \\\\
     \textcolor{white}{A^{(P)}_{\rm RS}(t,\boldsymbol{r})=}&-\dfrac{P-1}{2}\b J_0 \m^P+\dfrac{\beta'\,^2J^2}{4}\left(1-P\q^{P-1}+(P-1)\q^P\right)\,.
\end{array}
\label{hop_agssolution}
\end{equation}
where the order parameters are ruled by the following self-consistency equations
\begin{equation}
    \begin{array}{lll}
         \m=\left\langle\tanh{\left[\dfrac{P}{2}\beta J_0 \m^{^{P-1}} +z \beta J\sqrt{\dfrac{P }{2}\q^{^{P-1}}}\right]}\right\rangle_z \,,
         \\\\
         \q=\left\langle\tanh{}^2{\left[\dfrac{P}{2}\beta J_0 \m^{^{P-1}} +z \beta J\sqrt{\dfrac{P }{2}\q^{^{P-1}}}\right]}\right\rangle_z \,.
    \end{array}
    \label{eq:self_transRS}
\end{equation}
 \end{corollary}

We stress that the expression found with transport equation method in Corollary (\ref{cor:A_RS_trans}) is the same found with Guerra's interpolation in Proposition (\ref{SK_P_quenched}).

\subsubsection{1-RSB solution}\label{ssec:HRSB1}

In this subsection we turn to the solution of the P-spin glass model via the generalized  broken-replica interpolating technique, restricting the description at the first step of RSB.
\newline

The definition of two-replica overlap $q$ distribution is the same as Guerra's interpolating technique (See eq. \eqref{def:HM_RSB}).

Following the same route pursued in the previous sections, we need an interpolating partition function $\mathcal Z$ and an interpolating quenched pressure $\mathcal A$,  that are defined hereafter.
\begin{definition}
Given the interpolating parameters $\bm r = (x^{(1)}, x^{(2)},w), t$ and the i.i.d. auxiliary fields $\{z_i^{(1)}, z_i^{(2)}\}_{i=1,...,N}$, with $z_i^{(1,2)} \sim \mathcal N(0,1)$ for $i=1, ..., N$, we can write the 1-RSB interpolating partition function $\mathcal Z_N(t, \boldsymbol r)$ for the P-spin glass model (\ref{eq:SK_hbare}) recursively, starting by

\begin{equation}
\label{eqn:Z2_TrasP}
\footnotesize
\begin{array}{lll}
     \mathcal{Z}^{(P)}_2(t,\boldsymbol{r}) &:=& \sommaSigma{\boldsymbol \sigma} \exp{}\Bigg[t\dfrac{\b J_0 N}{2} m^P(\boldsymbol \sigma)+w N\,m(\boldsymbol \sigma)+
     \\\\
     & &+\sqrt{t}\b J\sqrt{\dfrac{1 }{2N^{P-1}}}\SOMMA{i_1,\cdots,i_{_{P}}=1}{N,\cdots,N}z_{i_1\cdots,i_{_{P}}}\sigma_{i_1}\cdots\sigma_{i_{P}}+
     \\\\
     & &+\SOMMA{a=1}{2}\left(\sqrt{x^{(a)}}\SOMMA{i=1}{N} z_i^{(a)}\sigma_i\right)\Bigg].
\end{array}
\end{equation}

Averaging out the fields recursively, we define
\begin{align}
\label{eqn:Z1_1RSB}
\mathcal Z_1(t, \bm r) \coloneqq& \mathbb E_2 \left [ \mathcal Z_2(t, \bm r)^\theta \right ]^{1/\theta} \\
\label{eqn:Z0_1RSB_trans}
\mathcal Z_0(t, \bm r) \coloneqq&  \exp \mathbb E_1 \left[ \ln \mathcal Z_1(t, \bm r) \right ] \\
\mathcal Z_N(t, \boldsymbol r) \coloneqq& \mathcal Z_0(t, \bm r) ,
\end{align}
where with $\mathbb E_a$ we mean the average over the variables $z_i^{(a)}$'s for $a=1, 2$, and with $\mathbb{E}_0$ we shall denote the average over the variables $z_{i_1\cdots,i_{_{P}}}$'s.
\end{definition}

The definition of 1RSB interpolating pressure at finite volume $N$ and in the thermodynamic limit is the same as Guerra's interpolating technique \eqref{def:interpPressRSB}. The notation for the generalized average are the analogous as Guerra's interpolation. 
\par\medskip

The next step is building a transport equation for the interpolating quenched pressure, for which we preliminary need to evaluate the related partial derivatives, as discussed in the next
\begin{lemma} \label{lemma:4}
The partial derivative of the interpolating quenched pressure with respect to a generic variable $\rho$ reads as

\begin{equation}
\label{eqn:partialrA}
\frac{\partial }{\partial \rho} \mathcal A_N^{(P)}(t, \bm r)=\frac{1}{N} \mathbb E_0  \mathbb E_1  \mathbb E_2 \left[\mathcal W_2 \omega \big( \partial_\rho \mathcal B(\bm \sigma; t, \bm r) \big)\right].
\end{equation}
In particular,
\begin{align}
\label{eqn:partialtA}
\frac{\partial }{\partial t} \mathcal A^{(P)}_N =&\frac{\beta^{'} J_0}{2}\langle m^P \rangle+ \frac{{\beta^{'}}^2 J^2}{4}\big (1 -(1-\theta)\langle q_{12}^{P} \rangle_2-\theta\langle q_{12}^{P} \rangle_1 \big) \\
\label{eqn:partialx1A}
\frac{\partial }{\partial x^{(1)}} \mathcal A^{(P)}_N =& \frac{1}{2}\big ( 1-(1-\theta)\langle q_{12} \rangle_2-\theta\langle q_{12} \rangle_1 \big) \\
\label{eqn:partialx2A}
\frac{\partial }{\partial x^{(2)}} \mathcal A^{(P)}_N =& \frac{1}{2}\big ( 1 -(1-\theta)\langle p_{12} \rangle_2
\big)\\
\label{eqn:partialwA}
\frac{\partial }{\partial w}  \mathcal A^{(P)}_N =&  \langle m \rangle
\end{align}
\end{lemma}
\begin{proof}
The proof is pretty lengthy and basically requires just standard calculations similar to RS case, since we omit it. Here we just prove that, in  complete generality
\begin{align}
\label{eqn:partialrA1}
\frac{\partial }{\partial \rho}  \mathcal A^{(P)}_N(t, \bm r)=& \frac{1}{N} \mathbb E_0 \mathbb E_1 \bigg [\partial_\rho \ln\mathcal Z_1 \bigg] \nonumber \\
=&\frac{1}{N} \mathbb E_0 \mathbb E_1 \bigg[\frac{1}{\theta}\frac{1}{\mathcal Z_1} \big[ \mathcal Z_2^\theta \big]^{1/\theta-1} \mathbb E_2 \big[\partial_\rho\mathcal Z_2^\theta \big] \bigg] \nonumber \\
=&\frac{1}{N} \mathbb E_0 \mathbb E_1 \mathbb E_2\bigg[\frac{\mathcal Z_2^\theta}{\mathbb E_2 \mathcal Z_2^\theta }\frac{\partial_\rho \mathcal Z_2}{\mathcal Z_2} \bigg]  \nonumber\\
=&\frac{1}{N} \mathbb E_0 \mathbb E_1 \mathbb E_2 \bigg[\mathcal W_2\frac{\partial_\rho \mathcal Z_2}{\mathcal Z_2} \bigg].
\end{align}
\end{proof}

\begin{proposition}
\label{prop:9_G}
The streaming of the 1-RSB interpolating quenched pressure obeys, at finite volume $N$, a standard transport equation, that reads as
\begin{align}
\label{eqn:transportequation}
\frac{d\mathcal A^{(P)}}{dt}&=\partial_t \mathcal A^{(P)}+\dot x^{(1)}\partial_{x_1} \mathcal A^{(P)} +\dot x^{(2)}\partial_{x_2} \mathcal A^{(P)}+\dot w \partial_{w} \mathcal A^{(P)} = S(t, \bm r) + V_N(t, \bm r)
\end{align}
where
\begin{align}
\label{eqn:f}
&S(t, \bm r)  \coloneqq \frac{{\beta^{'}}^2 J^2}{4} \left[ 1 -P\q_2^{P-1}+(P-1)\q_2^P-\theta (P-1)(\q_2^P  -  \q_1^P)\right] - (P-1)\frac{\b J_0 }{2}\m^P\,, \\
\label{eqn:V}
&V_N(t, \bm r) \coloneqq  \frac{{\beta^{'}}^2 J^2}{4}\left\{ (\theta -1) \sum_{k=2}^{P} \binom{P}{k} \bar{q}_2^{P - k} \langle (\Delta q_2)^k \rangle_2  - \theta \sum_{k=2}^{P} \binom{P}{k} \bar{q}_1^{P - k} \langle (\Delta q_1)^k \rangle_2 \right\} \notag \\
&\textcolor{white}{V_N(t, \bm r) \coloneqq }+ \frac{\b J_0}{2}\sum_{k=2}^P \binom{P}{k} \langle (\Delta m)^k \rangle\,,
\end{align}
\normalsize
where, as usual, we use $\Delta X_a=X-\bar{X}_a$.
\end{proposition}

\begin{proof}
Similar to the case RS (\ref{potential_m1}-\ref{potential_pq}), we have, for $a=1,2$
\begin{align}
\label{eq:q12_RSB}
    \langle q_{12}^{P}\rangle_a &= \sum_{k=1}^{P} \begin{pmatrix}P\\ k\end{pmatrix}\bar{q}_a^{P-k} \langle (\Delta q_a)^k \rangle_a +  \q_a^P (1-P) + P \q_a^{P-1} \langle q_{12} \rangle_a\,.
\end{align}
\normalsize
Now, starting to evaluate explicitly $\dt \mathcal{A}^{(P)}_N$ by using (\ref{eqn:partialx1A} - \ref{eqn:partialwA}) we write
\begin{align}
\dt \mathcal{A}^{(P)}_N=& \frac{\b J_0}{2} \left[ \sum_{k=2}^P \binom{P}{K} \langle (\Delta m)^k \rangle \bar{m}^{P-k} + \bar{m}^P (1-P) + P\bar{m}^{P-1}\langle m \rangle \right] \notag \\
&+ \frac{{\beta^{'}}^2 J^2}{4} \Bigg\{1+ (\theta-1)\left[  \sum_{k=1}^{P} \begin{pmatrix}P\\ k\end{pmatrix}\bar{q}_2^{P-k} \langle (\Delta q_2)^k \rangle_2 +  \q_2^P (1-P) + P \q_2^{P-1} \langle q_{12} \rangle_2\right] \notag \\
&- \theta \left[ \sum_{k=1}^{P} \begin{pmatrix}P\\ k\end{pmatrix}\bar{q}_1^{P-k} \langle (\Delta q_1)^k \rangle_a +  \q_1^P (1-P) + P \q_1^{P-1} \langle q_{12} \rangle_1\right]\Bigg\} = \notag \\
=& V_N(t, \bm r) + S(t, \bm r) + \frac{\b J_0 P}{2}\bar{m}^{P-1} \partial_w \mathcal{A}^{(P)}_N + \frac{\beta'\,^2 J^2 P }{2} \bar{q}_1^{P-1}\partial_{x^{(1)}}\mathcal{A}^{(P)}_N\notag \\
&+\frac{\beta'\,^2 J^2 P }{2}(\bar{q}_2^{P-1} - \bar{q}_1^{P-1})\partial_{x^{(2)}}\mathcal{A}^{(P)}_N \,.
\end{align}
\normalsize
Thus, by placing
\begin{align}
 \label{eqn:dotx1}
\dot{x}^{(1)}&= -\frac{\beta'\,^2 J^2 P }{2} \bar{q}_1^{P-1}\,, \\
\label{eqn:dotx2}
\dot{x}^{(2)}&= - \frac{\beta'\,^2 J^2 P }{2}(\bar{q}_2^{P-1} - \bar{q}_1^{P-1}) \,,\\
\dot{w} &= - \frac{\b J_0 P}{2}\bar{m}^{P-1}\,. \label{eq:dotw}
\end{align}
thus, we reach the thesis.

\end{proof}

\begin{remark} \label{r:above1}
In the thermodynamic limit, in the 1RSB scenario, we have
\begin{align}
\label{eqn:thlimaveragem}
 \lim_{N\rightarrow \infty} \langle (m - \bar m)^2 \rangle =& 0\,,\\
 \lim_{N\rightarrow \infty} \langle (q_{12}-\bar q_i)^2 \rangle_i=& 0; \: \: \: i=1,2\,.
\end{align}
Similar to the RS approximation, in the thermodynamic limit we have that the central moments greater than 2 tend to zero; in this way
\begin{equation} \label{eq:V0_HRSB}
\lim_{N \to \infty} V_N(t, \bm r) = 0.
\end{equation}
\end{remark}

Exploiting Remark \ref{r:above1} we can prove the following
\begin{proposition} \label{propHRSB_G}
The transport equation associated to the interpolating pressure function $\mathcal A^{(P)}(t, \bm r)$, in the thermodynamic limit and under the 1RSB assumption, can be written as
\begin{align}
    \label{transeqfinale}
    &\partial_t \mathcal A^{(P)}+\dot x^{(1)}\partial_{x_1} \mathcal A^{(P)} +\dot x^{(2)}\partial_{x_2} \mathcal A^{(P)} +\dot z \partial_{z} \mathcal A^{(P)}+\dot w \partial_{w} \mathcal A^{(P)} \notag \\
    &= \frac{\beta'\,^2 J^2}{4} \Bigg[ 1 -P\q_2^{P-1}+(P-1)\q_2^P- \theta(P-1)( \q_2^P  - \q_1^P)\Bigg] - \frac{\b J_0 \m^P (P-1)}{2}
\end{align}
\normalsize
whose explicit solution is given by
\begin{align}
    &\mathcal A^{(P)} = \log 2 + \frac{1}{\theta} \mathbb{E}_1 \log \mathbb{E}_2 \cosh^\theta \left( w_0 + z^{(1)} \sqrt{x_0^{(1)}} + z^{(2)} \sqrt{x_0^{(2)}}\right) \notag \\
    &+t \left\{ \frac{\beta'\,^2 J^2}{4} \Bigg[ 1 -P\q_2^{P-1}+(P-1)\q_2^P- \theta(P-1)( \q_2^P  - \q_1^P)\Bigg] - \frac{\b J_0 \m^P (P-1)}{2} \right\}
\end{align}
\end{proposition}
\begin{proof}
If we consider \eqref{eq:V0_HRSB} we have \eqref{transeqfinale}
We compute the transport equation trough characteristic method: 
\begin{equation}
\label{eqn:solutionzeroV}
\mathcal A^{(P)}_{\textrm{1RSB}}(t, \bm r)=\mathcal A^{(P)}_{\textrm{1RSB}}(0, \bm{r}-\bm{\dot r}t)+S_{\textrm{1RSB}}(t, \bm r )t.
\end{equation}
By putting (\ref{eqn:dotx1})-(\ref{eq:dotw}) into (\ref{eqn:transportequation}) we find
\begin{align}
\label{eqn:solutionzeroV1}
\mathcal A^{(P)}(t, \bm r)=&  \mathcal A^{(P)}(0, \bm r_0)+t \left\{ - \frac{\b J_0 \m^P (P-1)}{2} \right. \notag \\
&\left.+\frac{\beta'\,^2 J^2}{4} \Bigg[ 1 -P\q_2^{P-1}+(P-1)\q_2^P- \theta(P-1)( \q_2^P  - \q_1^P)\Bigg]  \right\}
\end{align}
\normalsize
where $r_0$ can be obtained by using the equation of motion
\begin{equation}
\label{eqn:linearmotion_trans}
\bm r = \bm r_0 + \dot{\bm r} t
\end{equation}
where the velocities are defined in (\ref{eqn:dotx1})-(\ref{eq:dotw}). Then, all we have to compute is $\mathcal A^{(P)}_0(0, \bm r_0)$, that can be easily done because at $t=0$ the two body interaction vanishes and it can be written as
\begin{align}
\label{eqn:A0fin}
\mathcal A^{(P)}_0(0, \bm r_0)=& \log 2 + \frac{1}{\theta} \mathbb{E}_1 \log \mathbb{E}_2 \cosh^\theta \left( w_0 + z^{(1)} \sqrt{x_0^{(1)}} + z^{(2)} \sqrt{x_0^{(2)}}\right).
\end{align}
We omit the computation since it is similar to one-body term in Guerra's interpolating scheme, 1RSB assumption.
\newline
Then, putting together (\ref{eqn:solutionzeroV1})-(\ref{eqn:A0fin}) 
and (\ref{eqn:dotx1})-(\ref{eq:dotw}), we finally achieve an explicit expression for the interpolating pressure of the P-spin glass model in the 1RSB approximation.
\end{proof}

To sum up, we have the following main theorem for the 1RSB scenario
\begin{theorem}
The 1-RSB quenched pressure for P-spin glass model, in the thermodynamic limit, reads as
\begin{equation}
\label{eqn:hopfieldAfinal}
\begin{array}{lll}
     \mathcal A^{(P)}(\beta)=& \log 2 + \dfrac{1}{\theta} \mathbb{E}_1 \log \mathbb{E}_2 \cosh^\theta g(\bm z, \m) -(P-1)\dfrac{\b J_0}{2}\m^{P}
     \\\\
&+ \dfrac{\beta'\,^2 J^2}{4} \left[ 1 + (\theta-1) \q_2^P (1-P) - \theta \q_1^P (1-P) - P \q_2^{P-1} \right]
\end{array}
\end{equation}
where 
\begin{align}
\small
    g(\bm z, \m) &= \b J_0 \dfrac{P}{2}\m^{P-1}+\b J z^{(1)} \sqrt{\frac{P}{2}\q_1^{P-1}}+ \b J z^{(2)} \sqrt{\frac{P}{2}(\q_2^{P-1}-\q_1^{P-1})}\,.
    \label{def_g_trans}
\end{align}

The order parameters are ruled by
\begin{align}
&\bar{m}= \mathbb{E}_1 \left[ \frac{\mathbb{E}_2 \cosh^\theta g(\bm z,\bar{m})\tanh g(\bm z,\bar{m})}{\mathbb{E}_2 \cosh^\theta g(\bm z,\bar{m})}\right] \\
&\bar{q}_1 = \mathbb{E}_1 \left[ \frac{\mathbb{E}_2 \cosh^\theta g(\bm z,\bar{m})\tanh g(\bm z,\bar{m})}{\mathbb{E}_2 \cosh^\theta g(\bm z,\bar{m})}\right]^2 \\
&\bar{q}_2 = \mathbb{E}_1 \left[ \frac{\mathbb{E}_2 \cosh^\theta g(\bm z,\bar{m})\tanh^2 g(\bm z,\bar{m})}{\mathbb{E}_2 \cosh^\theta g(\bm z,\bar{m})}\right]
\end{align}
with $g(\bm z, \m)$ is the same defined in \eqref{def_g_trans}.

\end{theorem}
\begin{proof}
By taking $\bm r= \boldsymbol 0$ and $t=1$ we find the P-spin glass pressure in the 1RSB approximation.
\end{proof}
We stress that the expression in \eqref{eqn:hopfieldAfinal} is the same found with Guerra's interpolating scheme in 1RSB assumption. 
%
%
%


\section{Gaussian P-Spin Glass}
In this section we face with Gaussian P-spin glass model, namely spin glass model with Gaussian spins with $P$-wise interactions. Several works describe the principal aspects of this model, either in our assumption \cite{BenArous,Bovier3} or in the so-called spherical version, where the Gaussian spin are on a spherical surface\cite{crisanti,PRL, TalaP}. 
We stress that two assumptions on Gaussian spins are equivalent, as shown in \cite{AdrianoGauss}. 
We deepen the model through Guerra's interpolating technique and transport equation for RS and 1RSB assumptions. 

\subsection{Generalities}
\label{GaussSection}

\begin{definition} 
\label{def:pspinhamGaus} 
Let $P\in\mathbb{N}$ and $z_i\in\mathbb{R},\;i=1,2,\cdots,N$ be a configuration of N gaussian spins, the Hamiltonian of the P-spin Gaussian model is defined as
	\begin{equation}
	H_N^{(P)}(\boldsymbol z| \boldsymbol J) \coloneqq -\dfrac{1}{P!}\SOMMA{i_1,\cdots,i_P=1}{N,\cdots,N}\,\mathcal{J}_{i_1,\cdots,i_P}z_{i_1}\cdots z_{i_P}
	\label{eq:G_hbare}
	\end{equation}
where the P-wise quenched couplings $\boldsymbol{\mathcal{J}} = \lbrace \mathcal{J}_{i_1,\cdots,i_P}\rbrace_{i_1,\cdots,i_P=1,\cdots,N}$ are given by
\begin{equation}
    \mathcal{J}_{i_1,...,i_P}\coloneqq\,\sqrt{\dfrac{2}{N^{P-1}}}J_{i_1,\cdots,i_P}
\end{equation}
with $J_{i_1,\cdots,i_P}$ i.i.d. standard random variables drawn from $P(J_{i_1,\cdots,i_P}) = N(0; 1)$.
\end{definition}

\begin{definition}
\label{G_BareZ}
The partition function related to the Hamiltonian \eqref{eq:G_hbare} is given by
\begin{equation}
    \begin{array}{lll}
    \label{eq:G_BareZ}
	Z_N(\beta,\boldsymbol J) &\coloneqq \displaystyle\int d\mu(z) \exp \left[ -\beta H_N^{(P)}(z | \boldsymbol J)\right]=\mathbb{E}_z \,e ^{-\beta H_N^{(P)}(z | \boldsymbol J)}\,,
    \end{array}
\end{equation}
	where, $d\mu(z)=\prod_i (2\pi)^{-1/2}e^{-\sum_i z_i^2/2}$ and $\beta\in \mathbb R^+$ is the inverse temperature in proper units.
	
	As early pointed out for
instance in \cite{Berlin}, unfortunately these kind of models need to be regularized; in
fact, the right side of \eqref{eq:G_BareZ} is not always well defined as the P-wise interactions bridges soft spins which are both Gaussian distributed. It will be clear soon that a good definition is
\begin{equation}
    \begin{array}{lll}
    \label{eq:G_BareZ_1}
	Z_N(\beta,\boldsymbol J,\lambda) &\coloneqq \mathbb{E}_z \,\exp\left[ -\beta H_N^{(P)}\left(z | \boldsymbol J\right)-\left(\dfrac{\beta}{P!}\right)^2\dfrac{1}{N^{P-1}}\left(\SOMMA{i=1}{N}z_i^2\right)^P+\dfrac{\lambda}{2}\SOMMA{i=1}{N}z_i^2\right]\,,
    \end{array}
\end{equation}
where the first additional term is needed for convergence of the integral over the Gaussian measure $d\mu(z)$. Instead, the new parameter $\lambda$, within the last term of \eqref{eq:G_BareZ_1}, is just to modify the variance of the soft spins, as in several applications this can sensibly vary.
\end{definition}

We introduce the \emph{Boltzmann average} induced by the partition function (\ref{eq:G_BareZ}), denoted with $\omega_{\boldsymbol J}$ and, for an arbitrary observable $O(\boldsymbol \sigma)$, defined as
	\begin{equation}
	\omega_{\boldsymbol J} (O (\boldsymbol \sigma)) : = \frac{\mathbb{E}_z O(\boldsymbol \sigma) e^{- \beta H_N(\boldsymbol \sigma| \boldsymbol J)}}{Z_N(\beta, \boldsymbol J,\lambda)}.
	\end{equation}
	This can be further averaged over the realization of the $J_{i_1,\cdots,i_P}$'s (also referred to as \emph{quenched average}) to get
	\begin{equation}
	  \langle O(\boldsymbol \sigma) \rangle \coloneqq \mathbb{E} \omega_{\boldsymbol J} (O(\boldsymbol \sigma)).  
	\end{equation}

\begin{definition}The intensive quenched pressure of the Gaussian P-spin glass model (\ref{eq:G_hbare}) is defined as
\begin{equation}
\label{PressureDef_G}
\mathcal A_N(\beta,\lambda) \coloneqq \frac{1}{N} \mathbb{E} \ln \mathcal Z_N(\beta, \boldsymbol J,\lambda),
\end{equation}
and its thermodynamic limit, assuming its existence, is referred to as
\begin{equation}
\mathcal A(\beta,\lambda) \coloneqq \lim_{N \to \infty} \mathcal A_N(\beta,\lambda).
\label{eq:quenched_pressure_def_G}
\end{equation}
\end{definition}

In order to solve the model we want to find out an explicit expression for the quenched pressure \eqref{eq:quenched_pressure_def_G} in terms of the natural order parameter of the theory, namely the two-replica overlap $q_{12}$, defined in the following

\begin{definition} The order parameter used to describe the macroscopic behavior of the model is the two-replica overlap, introduced as
\begin{equation}
\label{q_G}
q_{12} \coloneqq \frac{1}{N}\sum_{i=1}^N z_i^{(1)}z_i^{(2)}\,.
\end{equation}
\end{definition}

\subsection{Resolution via Guerra's interpolation}

Mirroring P-spin glass model, the purpose of this section is to solve the Gaussian P-spin glass model via Guerra's interpolating technique. To do so, we recover the expressions of statistical pressure (which is equivalent to free energy) and self consistency equations in the approximation of replica symmetry (RS) and first step of replica symmetry breaking (RSB).

\subsubsection{RS solution}
Using the same definition for $q_{12}$ presented for the P-spin glass model (def. \ref{defn: RSassumption}), we can introduce the Guerra's interpolating partition function of the Gaussian P-spin glass model.

\begin{definition} 
Given the interpolating parameter $t \in [0,1]$, $A,\ B \in \mathbb{R}$  and $\tilde{J}_i \sim \mathcal{N}(0,1)$ for $i=1, \hdots , N$ standard Gaussian variables i.i.d., the Guerra's interpolating partition function is given as 
\begin{equation}
\footnotesize
\begin{array}{lll}
     \mathcal{Z}^{(P)}_N(t) &\coloneqq& \mathbb{E}_z\exp{}\Bigg[\b\sqrt{t}\sqrt{\dfrac{1}{2N^{P-1}}}\SOMMA{i_1,\cdots,i_P=1}{N,\cdots,N}J_{i_1,\cdots,i_P}z_{i_1}\cdots z_{i_P}-t\dfrac{\beta'\,^2}{4N^{P-1}}\left(\SOMMA{i=1}{N}z_i^2\right)^P
     \\\\
     & &+ A\sqrt{1-t}\SOMMA{i=1}{N}\tilde{J}_i z_i+(1-t)\dfrac{B}{2}\SOMMA{i=1}{N}z_i^2+\dfrac{\lambda}{2}\SOMMA{i=1}{N}z_i^2\Bigg]\,,
     \label{def:partfunct_GuerraRS_G}
\end{array}
\end{equation}
where $\b=2\beta/P!$.
\end{definition}
\begin{definition} The interpolating pressure for the Gaussian P-spin glass model (\ref{eq:G_hbare}), at finite $N$, is introduced as
\begin{eqnarray}
\mathcal{A}^{(P)}_N(t) &\coloneqq& \frac{1}{N} \mathbb{E} \left[  \ln \mathcal{Z}^{(P)}_N(t)  \right],
\label{G_GuerraAction}
\end{eqnarray}
where the expectation $\mathbb E$ is now meant over $J_{i_1,\cdots,i_P}$ and $\tilde{J}_i$, in the thermodynamic limit,
\begin{equation}
\mathcal{A}^{(P)}(t) \coloneqq \lim_{N \to \infty} \mathcal{A}^{(P)}_N(t).
\label{G_GuerraAction_TDL}
\end{equation}
By setting $t=1$ the interpolating pressure recovers the original one (\ref{PressureDef_G}), that is $\mathcal A_N^{(P)} (\beta,\lambda) = \mathcal{A}^{(P)}_N(t=1)$.
\end{definition}

\begin{remark}
	The interpolating structure implies an interpolating measure, whose related Boltzmann factor reads as
	\begin{equation}
	\mathcal B (\boldsymbol z; t )\coloneqq  \exp \left[ \beta  H (\boldsymbol z; t) \right];
	\end{equation}

In this way the partition function is written as $\mathcal Z_N(t) =  \mathbb{E}_z \mathcal B (\boldsymbol z; t)$  .\\
A generalized average follows from this generalized measure as
\beq
	\omega_{t} (O (\boldsymbol z )) \coloneqq   \mathbb{E}_z O (\boldsymbol z ) \mathcal B (\boldsymbol z; t)
	\eeq
	and
\beq
\langle O (\boldsymbol z ) \rangle_{t}  \coloneqq \mathbb E [ \omega_{t} (O (\boldsymbol z)) ].
\eeq
where $ \mathbb E$ denotes the average over $\tilde{J}_{i_1,\cdots,i_P}$ and $\lbrace J_i\rbrace_{i=1,\cdots,N}$.

Of course, when $t=1$ the standard Boltzmann measure and related average is recovered.
Hereafter, in order to lighten the notation, we will drop the subindices $t$.
\end{remark}
\begin{lemma} 
The $t$ derivative of interpolating pressure is given by 
\begin{equation}
    \begin{array}{lll}
         \dfrac{d \mathcal{A}^{(P)}(t)}{d t}\coloneqq & -\dfrac{\beta'\,^2}{4}\left(\l q_{12}\r^P-\dfrac{2A^2}{\beta'^2}\l q_{12}\r\right)-\dfrac{1}{2N}\left(A^2+B\right)\SOMMA{i=1}{N}\l z_i^2\r
    \end{array}
    \label{eq:streaming_RS_Guerra_G}
\end{equation}
\end{lemma}

\begin{proof}
Deriving equation \eqref{G_GuerraAction} with respect to $t$, we get 
\begin{equation}
\begin{array}{lll}
     \dfrac{d \mathcal{A}^{(P)}(t)}{d t}&=&\dfrac{\b }{2N \sqrt{t}}\sqrt{\dfrac{1}{ 2N^{P-1}}}\mathbb{E}\left[\sum\limits_{\boldsymbol{i}}J_{\boldsymbol{i}}\omega(z_{i_1}\cdots z_{i_{P}})\right]-\dfrac{\beta'\,^2}{4N^P}\left[\sum\limits_i\omega(z_i^2)\right]^P+
     \\\\
     &&-\dfrac{1}{2N\sqrt{1-t}}\mathbb{E}\left[A\sum\limits_{i} \tilde{J}_{i}\omega(z_{i})\right]-\dfrac{1}{2N}\mathbb{E}\left[B\sum\limits_{i}\omega(z_{i}^2)\right]\,.
\end{array}
\label{eq:proof_streaming_Guerra_G}
\end{equation}
Now, using the Stein's lemma (also known as Wick's theorem) on standard Gaussian variable $\tilde{J}_i$ and $J_{\boldsymbol{i}}$, we may rewrite the second and the third member of \eqref{eq:proof_streaming_Guerra_G} as
\begin{equation*}
\small
\begin{array}{lll}
     \dfrac{\b }{2N \sqrt{t}}\sqrt{\dfrac{1}{ 2N^{P-1}}}\mathbb{E}\left[\sum\limits_{\boldsymbol{i}}J_{\boldsymbol{i}}\omega(z_{i_1}\cdots z_{i_{P}})\right]-\dfrac{1}{2N\sqrt{1-t}}\mathbb{E}\left[A\sum\limits_{i} \tilde{J}_{i}\omega(z_{i})\right]=D_1+D_2\,.
\end{array}
\end{equation*}
Let’s investigate those three terms:
\begin{equation}
\small
\begin{array}{lll}
     D_1&=\dfrac{\b }{2N \sqrt{t}}\sqrt{\dfrac{1}{ 2N^{P-1}}}\mathbb{E}\left[\sum\limits_{\boldsymbol{i}}\partial_{J_{\boldsymbol{i}}}\omega(z_{i_1}\cdots z_{i_{P}})\right]
     \\\\
     &=\dfrac{\beta'\,^2}{4 N^{P}}\left(\sum\limits_{\boldsymbol{i}}\mathbb{E}\left[\omega((z_{i_1}\cdots z_{i_{P}})^2)\right]-\sum\limits_{\boldsymbol{i}}\mathbb{E}\left[\omega(z_{i_1}\cdots z_{i_{P}})^2\right]\right)
     \\\\
     &=\dfrac{\beta'\,^2}{4 N^{P}}\left(\sum\limits_{i}\mathbb{E}\left[\omega(z_{i}^2)\right]\right)^P-\dfrac{\beta'\,^2}{4}\l q_{12}^P\r\,;
\end{array}
\label{eq:D1_GuerraRS_G}
\end{equation}
\begin{equation}
\footnotesize
\begin{array}{lll}
     D_2&=
    -\dfrac{1}{2N\sqrt{1-t}}\mathbb{E}\left[A\sum\limits_{i} \partial_{\tilde{J}_{i}}\omega(z_{i})\right]
    \\\\
    &=-\dfrac{1}{2N}A^2\left(\sum\limits_{i}\mathbb{E}\left[ \omega(z^{2}_{i})\right]-\sum\limits_{i}\mathbb{E}\left[ \omega(z_{i})^{2}\right]\right)
     \\\\
     &=-\dfrac{1}{2}A^2\left[\dfrac{1}{N}\mathbb{E}\left[ \omega(z_{i})^{2}\right]-\l\qq\r\right]\,.
\end{array}
\label{eq:D2_GuerraRS_G}
\end{equation}
Rearranging together \eqref{eq:D1_GuerraRS_G} and \eqref{eq:D2_GuerraRS_G} we obtain the thesis.
\end{proof}

\begin{remark}
We stress that, for the RS assumption presented in Definition \eqref{defn: RSassumption}, we can use the relations \eqref{potential_m1} and \eqref{potential_pq} to fix the two constants as
\begin{equation}
    \begin{array}{lll}
         A^2=\dfrac{P}{2}\beta'\,^2 \q^{^{P-1}}\,,
         \\\\
         B=-A^2
    \end{array}
    \label{eq:constant_GuerraRS_G}
\end{equation}
the \eqref{eq:streaming_RS_Guerra_G} in the thermodynamical limit reads as
\begin{equation}
\begin{array}{lll}
     \dfrac{d \mathcal{A}^{(P)}(t)}{d t}&\coloneqq &(P-1)\dfrac{\beta'\,^2}{4}\q^P
\end{array}
\label{eq:streaming_RS_Guerra2_G}
\end{equation}
which is now independent of $t$.
\end{remark}

\begin{theorem}
\label{G_P_quenched}
In the thermodynamic limit ($N\to\infty$) and under RS assumption, applying the Fundamental Theorem of Calculus and using the suitable values of $A$ and $B$, we find the quenched pressure for the Gaussian P-spin glass model as
\begin{equation}
\label{eq:pressure_GuerraRS_G}
\begin{array}{lll}
     \mathcal{A}^{(P)}(\b,\lambda) &=& \dfrac{\frac{P}{2}\beta'\,^2\q^{P-1}}{2\left(1-\lambda+\frac{P}{2}\beta'\,^2\q^{P-1}\right)}-\dfrac{1}{2}\ln{\left[1-\lambda+\frac{P}{2}\beta'\,^2\q^{P-1}\right]}+
    \\\\
    &&+(P-1)\dfrac{\beta'\,^2}{4}\q^P\,.
\end{array}
\end{equation}
where the self-consistency equation which rule the order parameter are given by the resolution or the implicit equation
\begin{equation}
    \begin{array}{lll}
         \q=\dfrac{\frac{P}{2}\beta'\,^2\q^{P-1}}{\left(1-\lambda+\frac{P}{2}\beta'\,^2\q^{P-1}\right)^2} \,.
    \end{array}
    \label{eq:self_GuerraRS_G}
\end{equation}
\end{theorem}

\begin{proof}
Using the Fundamental Theorem of Calculus \eqref{eq:F_T_Calculus} and computing the one-body terms
\begin{equation}
\footnotesize
\begin{array}{lll}
     \mathcal{A}^{(P)}(t=0)&=&\dfrac{1}{N}\mathbb{E}_{\tilde{J}}\ln{}\Bigg\lbrace \mathbb{E}_z \exp{}\Bigg[A\,\SOMMA{i=1}{N} \tilde{J}_i z_i+\dfrac{1}{2}\left(B+\lambda\right)\SOMMA{i=1}{N}z_i^2\Bigg]\Bigg\rbrace\,,
     \\\\
         &=&\dfrac{A^2}{2\left(1-\lambda-B\right)}-\dfrac{1}{2}\ln{\left[1-\lambda-B\right]}
       \\\\
        &=&\dfrac{\frac{P}{2}\beta'\,^2\q^{P-1}}{2\left(1-\lambda+\frac{P}{2}\beta'\,^2\q^{P-1}\right)}-\dfrac{1}{2}\ln{\left[1-\lambda+\frac{P}{2}\beta'\,^2\q^{P-1}\right]}\,.
    \end{array}
    \label{eq:one_body_GuerraRS_G}
\end{equation}
Finally, putting \eqref{eq:streaming_RS_Guerra2_G} and \eqref{eq:one_body_GuerraRS_G} in \eqref{eq:F_T_Calculus}, we find \eqref{eq:pressure_GuerraRS_G}.

Extremizing the statistical pressure in \eqref{eq:pressure_GuerraRS_G} w.r.t. the order parameter we find equation \eqref{eq:self_GuerraRS_G}.
\end{proof}

A, briefly, deeper study of the interesting case of $P=2$ is done in Section \ref{sec:G_P2},  for a complete discussion we remand to \cite{AdrianoGauss}.

\begin{remark}
We highlight that we recover the same results in \cite{crisanti} in RS assumption computed via \textit{replica trick}. 

\par \medskip
Without going into details of implicit equation \eqref{eq:self_GuerraRS_G}, if we focus on $P > 2$ case, we stress that the solution $\q=0$ is always a saddle point; for $\q \neq 0$ we can see numerically that there is only an acceptable solutions, which maximize the quenched statistical pressure. For every details we remind to \cite{crisanti}. 
\end{remark}

\subsubsection{1RSB solution}

In this subsection we turn to the solution of the Gaussian P-spin glass model via the Guerra's interpolating technique, restricting the description at the first step of RSB.

Following the same route pursued in the previous sections, we need an interpolating partition function $\mathcal{Z}$ and an interpolating quenched pressure $\mathcal{A}$,  that are defined hereafter.
\begin{definition}
Given the interpolating parameter $t$ and the i.i.d. auxiliary fields $\lbrace \tilde{J}_i^{(1)}, \tilde{J}_i^{(2)}\rbrace_{i=1,...,N}$, with $\tilde{J}_i^{(1,2)} \sim \mathcal N(0,1)$ for $i=1, ..., N$ we can write the 1-RSB interpolating partition function $\mathcal Z_N(t)$ for the P-spin Gaussian model (\ref{eq:G_BareZ_1}) recursively, starting by

\begin{equation}
\begin{array}{lll}
     \mathcal{Z}^{(P)}_2(t) &:=& \mathbb{E}_z \exp \Bigg[ \sqrt{t}\b \sqrt{\dfrac{1 }{2N^{P-1}}}\SOMMA{i_1,\cdots,i_{_{P}}=1}{N,\cdots,N}J_{i_1\cdots,i_{_{P}}}z_{i_1}\cdots z_{i_{P}}-t\dfrac{\beta'\,^2}{4N^{P-1}}\left(\SOMMA{i=1}{N}z_i^2\right)^P
     \\\\
     & &+\sqrt{1-t}\SOMMA{a=1}{2}\left(A^{(a)}\SOMMA{i=1}{N} \tilde J_i^{(a)}z_i\right) +(1-t)\dfrac{B}{2}\SOMMA{i=1}{N}z_i^2+\dfrac{\lambda}{2}\SOMMA{i=1}{N}z_i^2\Bigg]
     \label{def:partfunct_GuerraRSB_G}
\end{array}
\end{equation}
where  the $J_{i_1\cdots,i_{_{P}}}$'s i.i.d. standard Gaussian. 
Averaging out the fields recursively, we define
\begin{align}
\label{eqn:Z1_G_trans}
\mathcal Z_1^{(P)}(t) \coloneqq& \mathbb E_2 \left [ \mathcal Z_2^{(P)}(t)^\theta \right ]^{1/\theta} \\
\label{eqn:Z0_G:trans}
\mathcal Z_0^{(P)}(t) \coloneqq&  \exp \mathbb E_1 \left[ \ln \mathcal Z_1^{(P)}(t) \right ] \\
\mathcal Z_N^{(P)}(t) \coloneqq & \mathcal Z_0^{(P)}(t) ,
\end{align}
where with $\mathbb E_a$ we mean the average over the variables $\tilde{J}_i^{(a)}$'s, for $a=1, 2$, and with $\mathbb{E}_0$ we shall denote the average over the variables $J_{i_1\cdots,i_{_{P}}}$'s.
\end{definition}

\begin{definition}
\label{def:interpPressRSB_G}
The 1RSB interpolating pressure, at finite volume $N$, is introduced as
\begin{equation}\label{AdiSK1RSB_G}
\mathcal A_N^{(P)} (t) \coloneqq \frac{1}{N}\mathbb E_0 \big[ \ln \mathcal Z_0^{(P)}(t) \big],
\end{equation}
and, in the thermodynamic limit, assuming its existence
\begin{equation}
\mathcal A^{(P)} (t) \coloneqq \lim_{N \to \infty} \mathcal A^{(P)}_N (t).
\end{equation}
By setting $t=1$, the interpolating pressure recovers the standard pressure (\ref{PressureDef_G}), that is, $\mathcal A^{(P)}_N(\beta,\lambda) = \mathcal A^{(P)}_N (t =1)$.
\end{definition}

\par \medskip
Now the next step is computing the $t$-derivative of the statistical pressure. In this way we can apply the fundamental theorem of calculus and find the solution of the original model. 

\begin{lemma}
\label{lemma:tderRSB_G}
The derivative w.r.t. $t$ of interpolating statistical pressure can be written as 
\begin{equation}
\begin{array}{lll}
      d_t \mathcal{A}^{(P)}_N =&  + \dfrac{\beta'\,^2}{4} \left[ (\theta-1) \langle q_{12}^P \rangle_2 - \theta \langle q_{12}^P \rangle_1 \right] - \frac{(A^{(1)})^2}{2} \left[\frac{1}{N} \sum_i \langle z_i^2 \rangle + (\theta-1) \langle q_{12} \rangle_2 - \theta \langle q_{12} \rangle_1 \right] \\\\
      &- \frac{(A^{(2)})^2}{2} \left[\frac{1}{N} \sum_i \langle z_i^2 \rangle + (\theta-1) \langle q_{12} \rangle_2 \right] - \frac{B}{2N} \sum_i \langle z_i^2 \rangle
\end{array}
\end{equation}
\end{lemma}
Since the proof is similar to SK spin glass case, we omit it. 

\begin{remark}
If we use \eqref{eq:q12_RSB}, in 1RSB assumption we fix the contasts as

\begin{align}
\label{value_psi_G}
    A_1 &= \beta'\,^2 \frac{P}{2} \q_1^{P-1}\,,\\
    A_2 &= \beta'\,^2 \frac{P}{2} (\q_2^{P-1} - \q_1^{P-1})\,,\\
    B &=  - \beta'\,^2 \frac{P}{2} \q_2^{P-1}\,.
\end{align}
Thus, the derivative w.r.t. $t$ in the thermodynamic limit is computed as 
\begin{align}
\label{dtA_1RSB_thermogauss_G}
    &d_t \mathcal{A}^{(P)}_N =\frac{\beta'\,^2}{4} (P-1) \q_2^P- \frac{\beta'\,^2}{4} \theta (P-1)( \q_2^P - \q_1^P)\,.
\end{align}
\end{remark}

\begin{proposition}
At finite size and under 1RSB assumption applying the Fundamental Theorem of Calculus and using the suitable values of $A^{(1)}, A^{(2)},B$, we find the quenched pressure for the Gaussian P-spin glass model as
\begin{equation}
\label{A_1RSB_finiteGauss_G}
\begin{array}{lll}
     \mathcal{A}^{(P)} =& \dfrac{\beta'\,^2}{4} (\theta-1) \left[ \SOMMA{K=2}{P} \langle q_{12}^k \rangle_2 \q_2^{P-k} + (1-P) \q_2^P \right] 
     \\\\&- \dfrac{\beta'\,^2 }{4}\theta\left[  \SOMMA{k=2}{P} \langle q_{12}^k \rangle_1 \q_1^{P-k} + (1-P) \q_1^P\right] -\dfrac{1}{2}\log \left( 1-\lambda + \beta'\,^2 \frac{P}{2} \q_2^{P-1} \right) 
    \\\\
    &+ \dfrac{1}{2\theta} \log \left( \dfrac{ 1-\lambda + \beta'\,^2 \frac{P}{2} \q_2^{P-1}}{ 1-\lambda + \beta'\,^2 \frac{P}{2} \q_2^{P-1} - \beta'\,^2 \frac{P}{2} \theta (\q_2^{P-1} - \q_1^{P-1})}\right)  
    \\\\
    &+\dfrac{\beta'\,^2 P}{4}\dfrac{\q_1^{P-1}}{ 1-\lambda + \beta'\,^2 \frac{P}{2} \q_2^{P-1} - \beta'\,^2 \theta \frac{P}{2}(\q_2^{P-1} - \q_1^{P-1})}
\end{array}
\end{equation}
\end{proposition}

\begin{proof}
We use the fundamental theorem of Calculus \eqref{eq:F_T_Calculus}. The last step we need is the computation of one-body term. Computing the Gaussian integral, we have
\begin{equation}
\label{eq:one body_soft}
\begin{array}{lll}
     \mathcal{A}^{(P)}(t=0)=& -\dfrac{1}{2}\log \left( 1-\lambda + \beta'\,^2 \frac{P}{2} \q_2^{P-1} \right) 
     \\\\
     &+\dfrac{1}{2\theta} \log \left( \dfrac{ 1-\lambda + \beta'\,^2 \frac{P}{2} \q_2^{P-1}}{ 1-\lambda + \beta'\,^2 \frac{P}{2} \q_2^{P-1} - (\b)^2 \frac{P}{2} \theta (\q_2^{P-1} - \q_1^{P-1})}\right)
     \\\\
    &+\dfrac{\beta'\,^2 P}{4}\dfrac{\q_1^{P-1}}{ 1-\lambda + \beta'\,^2 \frac{P}{2} \q_2^{P-1} - \beta'\,^2 \theta \frac{P}{2}(\q_2^{P-1} - \q_1^{P-1})}\,.
\end{array}
\end{equation}
In this way, we obtain the thesis. 
\end{proof}

Now we have the following
\begin{theorem}
In the thermodynamic limit, under 1RSB assumption, the expression of quenched statistical pressure for Gaussian P-spin glass model \eqref{eq:G_BareZ_1} is
\begin{equation}
    \begin{array}{lll}
         \mathcal{A}^{(P)}(\beta,\lambda)=& \dfrac{\beta'\,^2}{4} (P-1) \q_2^P- \dfrac{\beta'\,^2}{4} \theta (P-1)( \q_2^P - \q_1^P)
         \\\\
         &-\dfrac{1}{2}\log \left( 1-\lambda + \beta'\,^2 \frac{P}{2} \q_2^{P-1} \right) 
         \\\\
    &+ \dfrac{1}{2\theta} \log \left( \dfrac{ 1-\lambda + \beta'\,^2 \frac{P}{2} \q_2^{P-1}}{ 1-\lambda + \beta'\,^2 \frac{P}{2} \q_2^{P-1} - (\b)^2 \frac{P}{2} \theta (\q_2^{P-1} - \q_1^{P-1})}\right) 
    \\\\
    &+\dfrac{\beta'\,^2 P}{4}\dfrac{\q_1^{P-1}}{ 1-\lambda + \beta'\,^2 \frac{P}{2} \q_2^{P-1} - \beta'\,^2 \theta \frac{P}{2}(\q_2^{P-1} - \q_1^{P-1})}
    \end{array}
    \label{A_1RSB_finalissima_G}
\end{equation}
where the order parameters are ruled by the implicit equations
\begin{equation}
\footnotesize
    \begin{array}{lll}
         \bar{q}_1 = \dfrac{\beta'\,^2 \frac{P}{2} \q_1^{P-1}}{\left(1-\lambda + \beta'\,^2 \frac{P}{2} \q_2^{P-1} - \beta'\,^2\theta \frac{P}{2} (\q_2^{P-1}-\q_1^{P-1})\right)^2}
         \\\\
\bar{q}_2 = \q_1 + \dfrac{\beta'\,^2 \frac{P}{2} \left(\q_2^{P-1}-\q_1^{P-1}\right)}{\left(1-\lambda + \beta'\,^2 \frac{P}{2} \q_2^{P-1} - \beta'\,^2\theta \frac{P}{2} \left(\q_2^{P-1}-\q_1^{P-1}\right)\right)\left( 1-\lambda + \beta'\,^2 \frac{P}{2} \q_2^{P-1}\right)}
    \end{array}
    \label{eqn:self1rsbgauss}
\end{equation}
\end{theorem}

\begin{proof}

If we use \eqref{dtA_1RSB_thermogauss_G} and \eqref{eq:one body_soft} in \eqref{eq:F_T_Calculus} we recover \eqref{A_1RSB_finalissima_G}.

Extremizing the statistical pressure in \eqref{A_1RSB_finalissima_G} w.r.t. the order parameters we find the self-consistency equations.
\end{proof}

\begin{remark}
In the case of pairwise interaction ($P=2$) the 1RSB solution coincide with the RS one, as shown in Section \ref{sec:G_P2}. 

Moreover, we recover the same expression in \cite{crisanti}; we remind to it for the study of solutions behaviour of order parameters \eqref{eqn:self1rsbgauss}. 
\end{remark}


\subsection{Resolution via transport equation}
Now the aim is to show that, using the transport equation technique, already introduced in Section \ref{sec:transport1} for the RS assumption and for 1RSB in P-spin glass network, we are able to retrieve the same expression of statistical pressure above found by Guerra's Interpolation. 

\subsubsection{RS solution} \label{ssec:HRS_G}
As previous presented in Section \ref{sec:transport1}, now we introduce the interpolating partitions function for the Gaussian P-spin model.
\begin{definition} 
Given the interpolating parameter $t, x, w$, and $\tilde{J}_i \sim \mathcal{N}(0,1)$  standard Gaussian variables iid, the partition function is given as 
\begin{equation}
\footnotesize
\begin{array}{lll}
     \mathcal{Z}^{(P)}_N(t,\boldsymbol{r}) &\coloneqq& \mathbb{E}_z\exp{}\Bigg[\b\sqrt{t}\sqrt{\dfrac{1}{2N^{P-1}}}\SOMMA{i_1,\cdots,i_P=1}{N,\cdots,N}J_{i_1,\cdots,i_P}z_{i_1}\cdots z_{i_P}-t\dfrac{\beta'\,^2}{4N^{P-1}}\left(\SOMMA{i=1}{N}z_i^2\right)^P
     \\\\
     & &+ \sqrt{x}\SOMMA{i=1}{N}\tilde{J}_i z_i+\dfrac{w}{2}\SOMMA{i=1}{N}z_i^2+\dfrac{\lambda}{2}\SOMMA{i=1}{N}z_i^2\Bigg]\,,
     \label{def:partfunct_transpRS_G}
\end{array}
\end{equation}
where we set $\beta'=2 \beta/P!$.
\end{definition}
Similar to RS Guerra's interpolation, we can define the interpolating pressure, the Boltzmann factor and the generalized measure. 

\begin{lemma} 
The partial derivatives of the interpolating pressure (\ref{G_GuerraAction}) w.r.t. $t,x,w$ give the following expectation values:
	\bea
	\label{G_expvalsa}
	\frac{\partial \mathcal{A}^{(P)}_N}{\partial t} &=& -\dfrac{\beta'\,^2}{4}\l q_{12}^P\r,
	\\
	\label{G_expvalsmiddle}
	\frac{\partial \mathcal{A}^{(P)}_N}{\partial x}  &=& \dfrac{1}{2N}\SOMMA{i=1}{N}\l z_i^2\r-\dfrac{1}{2}\l q_{12}\r,\\
	\frac{\partial \mathcal{A}^{(P)}_N}{\partial w}  &=&  \dfrac{1}{2N}\SOMMA{i=1}{N}\l z_i^2\r.
	\label{G_expvalsb}
	\eea
\end{lemma}

\begin{proof}
We prove only the derivatives w.r.t. $t$, as long as the procedure is analogous for each parameter. The  partial  derivative  of  the  interpolating  quenched  pressure  with respect to $t$ reads as 
\begin{equation}
    \begin{array}{lll}
           \dfrac{\partial \mathcal{A}^{(P)}_N}{\partial t}&=&\dfrac{\b }{2N\sqrt{t}}\sqrt{\dfrac{1}{2N^{P-1}}}\sum\limits_{\boldsymbol{i}}\mathbb{E}\left[\omega(z_{i_1}\cdots z_{i_{P}})\right]-\dfrac{\beta'^2}{4N^P}\mathbb{E}\left[\omega\left(\SOMMA{i=1}{N}z_i^2\right)^P\right]\,.
    \end{array}
\end{equation}
Now, using the Stein's lemma \eqref{eqn:gaussianrelation2} on standard Gaussian variable $z_{i_1,\cdots,i_P}$
\begin{equation}
    \footnotesize
    \begin{array}{lll}
           \dfrac{\partial \mathcal{A}^{(P)}_N}{\partial t}&=& \dfrac{\beta'\,^2}{4N^{^{P}}}\left\lbrace\mathbb{E}\left[\omega((z_{i_1}\cdots z_{i_{P}})^2)\right]-\mathbb{E}\left[\omega(z_{i_1}\cdots z_{i_{P}})^2\right]-\mathbb{E}\left[\omega\left(\SOMMA{i=1}{N}z_i^2\right)^P\right]\right\rbrace-\dfrac{\beta'\,^2}{4}\l\qq^{P}\r\,.
    \end{array}
\end{equation}
\end{proof}

\begin{proposition} \label{prop:interp_transp_RS_G}
The interpolating pressure \eqref{def:partfunct_transpRS_G} at finite size obeys the following differential equation:
	\begin{equation}
	\frac{d \mathcal{A}^{(P)}_N}{dt} =  \pder{\mathcal{A}_N^{(P)}}{t} + \dot{x} \pder{\mathcal{A}^{(P)}_N}{x} +\dot w \pder{\mathcal{A}^{(P)}_N}{w}= S(t, \bm r)+V_N(t, \bm r),
	\label{hop_GuerraAction_DE_G}
	\end{equation}
	where we set 
	\begin{equation}
	    \begin{array}{lll}
	         \dot x = -\dfrac{P}{2}\beta'\,^2\q^{P-1}\,, && \dot w = -\dot x
	    \end{array}
	    \label{eq:deriv_r_trasp_RS_G}
	\end{equation}
	and
	\begin{equation}
	\begin{array}{lll}
	     S(t, \bm r) &\coloneqq& \dfrac{\beta'\,^2}{4}(P-1)\q^P,
	\end{array}
	\end{equation}
	\begin{equation}
	\begin{array}{lll}
	     V_N(t, \bm r) &\coloneqq& \dfrac{\beta'\,^2}{4}\SOMMA{k=2}{P}\begin{pmatrix}
	    P\\k
	\end{pmatrix} \q^{P-k}\l(\Delta q_{12})^k\r.
	\end{array}
	\label{potenziale-RS-G}
	\end{equation}
	\end{proposition}
	
	\begin{proof}
	Starting to evaluate explicitly $\dt \mathcal{A}_N$ by using (  \eqref{G_expvalsmiddle} - \eqref{G_expvalsb} and Definition \eqref{potential_m1}-\eqref{potential_pq}) we write
\begin{equation}
\footnotesize
    \begin{array}{lll}
         \dfrac{\partial}{\partial t}\mathcal{A}^{(P)}_N=-\dfrac{\beta'\,^2 }{4}\left(\SOMMA{k=2}{P} \begin{pmatrix}P\\k\end{pmatrix} \langle (q_{12}-\q)^k \rangle \q^{P-k} -\q^P (1-P) - P\q^{P-1}\langle q_{12} \rangle\right) 
         \\\\
         =V_N(t,\boldsymbol{r})+S(t,\boldsymbol{r})+\dfrac{\beta'\,^2}{2}P\q^{P-1}\dfrac{1}{2} \langle q_{12} \rangle
         \\\\
         =V_N(t,\boldsymbol{r})+S(t,\boldsymbol{r}) +\dfrac{\beta'\,^2}{2}P\q^{P-1}\Big(\frac{\partial \mathcal{A}^{(P)}_N}{\partial w} \Big)-\dfrac{\beta'\,^2 J}{2}P\q^{P-1}\Big(\frac{\partial \mathcal{A}^{(P)}_N}{\partial x} \Big)  
    \end{array}
\end{equation}
Thus, by placing $\dot{\boldsymbol{r}}=(\dot{x},\dot{w})$ as in \eqref{eq:deriv_r_trasp_RS_G} we reach the thesis.
	\end{proof}

	\begin{remark}
	\label{remark:Vthermolim_G}
	As shown for the P-spin glass model in the RS assumption (Remark \ref{remark:Vthermolim}), even for the Gaussian P-spin glass model, we can say that $V_N$, defined in \eqref{potenziale-RS-G}, vanishes. 
	\end{remark}
	
	\begin{proposition}
The transport equation associated to the interpolating pressure $\mathcal{A}_N(t, \bm r)$ in the thermodynamic limit and under the RS assumption is
	\begin{equation}
	\footnotesize
	\begin{array}{lll}
	     \pder{\mathcal{A}^{(P)}_{\textrm{RS}}}{t}+\dfrac{\beta'\,^2}{2}P\q^{P-1} \left(\pder{\mathcal{A}^{(P)}_{\textrm{RS}}}{x}\right)-\dfrac{\beta'\,^2}{2}P\q^{P-1} \left(\pder{\mathcal{A}^{(P)}_{\textrm{RS}}}{w}\right)= \dfrac{\beta'\,^2}{4}(P-1)\q^P,
	\end{array}
	\label{G_GuerraAction_RSDE}
	\end{equation}
	whose solution is given by
\begin{equation}
\small
\begin{array}{lll}
     \mathcal A^{(P)}_{\rm RS}(t,\boldsymbol{r})=&\frac{1}{2}-\ln{\left[1-\lambda-(w-\frac{P}{2}\b\q^{P-1}t)\right]}+\dfrac{x+\frac{P}{2}\beta'\,^2\q^{P-1}t}{2\left(1-\lambda-(w-\frac{P}{2}\b\q^{P-1}t)\right)}
     \\\\
     \textcolor{white}{A^{(P)}_{\rm RS}(t,\boldsymbol{r})=}&+\dfrac{\beta'\,^2}{4}t(P-1)\q^P\,.
\end{array}
\label{G_mechanicalsolution}
\end{equation}

\end{proposition}	

\begin{proof}
We can find the transport equation applying Remark \eqref{remark:Vthermolim_G} in Prop. \eqref{prop:interp_transp_RS_G}. 

We compute the solution using the characteristic method on the transport equation: 
\begin{align}
    \mathcal{A}^{(P)}_{RS}(t, \bm r) = \mathcal{A}^{(P)}_{RS}(0, \bm r-\dot{\boldsymbol{r}}t) + S(t, \bm r)t.
\end{align}
where $\dot{\boldsymbol{r}}=(\dot{x},\dot{w})$. Along the characteristics, the fictitious motion in the $(t,\boldsymbol{r})$ time-space is linear and returns
\begin{equation}
    \begin{array}{lll}
         x=x_0-\dfrac{P}{2}\beta'\,^2\q^{P-1}t &&  w=w_0+\dfrac{P}{2}\beta'\,^2 \q^{P-1}t
    \end{array}
\end{equation}
where $\boldsymbol{r}_0=(x_0,w_0)=(x(t=0),w(t=0))$. The Cauchy condition at $t=0$ is given by a direct computation at finite $N$ as
\begin{equation}
\begin{array}{lll}
     \mathcal A^{(P)}(0,\boldsymbol{r}-\dot{\boldsymbol{r}}t)&=&\mathcal A^{(P)}(0,\boldsymbol{r}_0)
     \\\\
     & =&\dfrac{1}{N}\mathbb{E}\Bigg\lbrace\mathbb{E}_z \exp{}\Bigg[ \dfrac{1}{2}\left(w_0+\lambda\right)\SOMMA{i=1}{N}z_i^2+\sqrt{x_0}\SOMMA{i=1}{N} \tilde{J}_i z_i\Bigg]\Bigg\rbrace
     \\\\
    & =&-\frac{1}{2}\ln{\left[1-\lambda-w_0\right]}+\dfrac{x_0}{2\left(1-\lambda-w_0\right)}\,.
\end{array}
\end{equation}
\end{proof}

Giving the suitable values of parameters we have the following
 \begin{corollary}
 \label{cor:A_RS_trans_G}
 The RS approximation of the quenched pressure, in the thermodynamic limit, for the Gaussian P-spin glass model is obtained by posing $t=1$ and $\bm r = \bm 0$ in (\ref{G_mechanicalsolution}), which returns
 \begin{equation}
\begin{array}{lll}
     \mathcal A^{(P)}_{\rm RS}(\b,\lambda)=&-\frac{1}{2}\ln{\left[1-\lambda+\frac{P}{2}\beta'\,^2\q^{P-1}\right]}+\dfrac{\frac{P}{2}\beta'\,^2\q^{P-1}}{2\left(1-\lambda+\frac{P}{2}\beta'\,^2\q^{P-1}\right)}
     \\\\
     \textcolor{white}{A^{(P)}_{\rm RS}(t,\boldsymbol{r})=}&+\dfrac{\beta'\,^2}{4}(P-1)\q^P\,.
\end{array}
\label{G_agssolution}
\end{equation}
where the self-consistency equation which rule the order parameter are 
\begin{equation}
    \begin{array}{lll}
         \q=\dfrac{\frac{P}{2}\beta'\,^2\q^{P-1}}{\left(1-\lambda+\frac{P}{2}\beta'\,^2\q^{P-1}\right)^2} \,.
    \end{array}
    \label{eq:self_GuerraRS_G_trans}
\end{equation}
 \end{corollary}

We stress that the expression found with transport equation method in Cor. (\ref{cor:A_RS_trans_G}) is the same found with Guerra's interpolation in Prop. (\ref{G_P_quenched}).

\subsubsection{1-RSB solution}\label{ssec:HRSB1_G}

In this subsection we turn to the solution of the Gaussian spin glass model through transport equation via the generalized  broken-replica interpolating technique, restricting the description at the first step of RSB.

The definition of two-replica overlap $q$ distribution is the same as Guerra's interpolation (See Def. \eqref{def:HM_RSB}).

\par\medskip
Following the same route pursued in the previous sections, we need an interpolating partition function $\mathcal Z$ and an interpolating quenched pressure $\mathcal A$,  that are defined hereafter.
\begin{definition}
Given the interpolating parameters $\bm r = (x^{(1)}, x^{(2)},w), t$ and the i.i.d. auxiliary fields $\{\tilde{J}_i^{(1)}, \tilde{J}_i^{(2)}\}_{i=1,...,N}$, with $\tilde J_i^{(1,2)} \sim \mathcal N(0,1)$ for $i=1, ..., N$, we can write the 1RSB interpolating partition function $\mathcal Z_N(t, \boldsymbol r)$ for the Gaussian P-spin glass model (\ref{eq:G_hbare}) recursively, starting by

\begin{equation}
\label{eqn:Z2_TrasP_G}
\footnotesize
\begin{array}{lll}
     \mathcal{Z}^{(P)}_2(t,\boldsymbol{r}) &:=& \mathbb{E}_z \exp{}\Bigg[\sqrt{t}\b \sqrt{\dfrac{1 }{2N^{P-1}}}\SOMMA{i_1,\cdots,i_{_{P}}=1}{N,\cdots,N}J_{i_1\cdots,i_{_{P}}}z_{i_1}\cdots z_{i_{P}}-t \dfrac{\beta'\,^2}{4N^{P-1}} \left( \SOMMA{i}{} z_i^2 \right)^P
     \\\\
     & &+\SOMMA{a=1}{2}\left(\sqrt{x^{(a)}}\SOMMA{i=1}{N} \tilde J_i^{(a)}z_i\right)+ \frac{w}{2} \sum_i z_i^2 + \frac{\lambda}{2} \sum_i z_i^2 \Bigg].
\end{array}
\end{equation}

Averaging out the fields recursively, we define
\begin{align}
\label{eqn:Z1_G}
\mathcal Z_1(t, \bm r) \coloneqq& \mathbb E_2 \left [ \mathcal Z_2(t, \bm r)^\theta \right ]^{1/\theta} \\
\label{eqn:Z0_G}
\mathcal Z_0(t, \bm r) \coloneqq&  \exp \mathbb E_1 \left[ \ln \mathcal Z_1(t, \bm r) \right ] \\
\mathcal Z_N(t, \boldsymbol r) \coloneqq& \mathcal Z_0(t, \bm r) ,
\end{align}
where with $\mathbb E_a$ we mean the average over the variables $\tilde{J}_i^{(a)}$'s for $a=1, 2$, and with $\mathbb{E}_0$ we shall denote the average over the variables $J_{i_1\cdots,i_{_{P}}}$'s.
\end{definition}

The definition of 1RSB interpolating pressure at finite volume $N$ and in the thermodynamic limit and the notation for the generalized average are the same as Guerra's interpolating technique \eqref{def:interpPressRSB_G}.
\par\medskip

The next step is building a transport equation for the interpolating quenched pressure, for which we preliminary need to evaluate the related partial derivatives, as discussed in the next
\begin{lemma} \label{lemma:44}
The partial derivative of the interpolating quenched pressure with respect to each variable is
\begin{align}
\label{eqn:partialtA_G1RSB}
\frac{\partial }{\partial t} \mathcal A^{(P)}_N =&\frac{\beta'\,^2 }{2} (\theta-1)\langle q_{12}^P \rangle_2-\frac{\beta'\,^2 }{2}\theta\langle q_{12}^P \rangle_1  \\
\label{eqn:partialx1A_G1RSB}
\frac{\partial }{\partial x^{(1)}} \mathcal A^{(P)}_N =& \frac{1}{2}\left( \frac{1}{N} \sum_i \langle z_i^2 \rangle -(1-\theta)\langle q_{12} \rangle_2-\theta\langle q_{12} \rangle_1 \right) \\
\label{eqn:partialx2A_G1RSB}
\frac{\partial }{\partial x^{(2)}} \mathcal A^{(P)}_N =& \frac{1}{2}\left( \frac{1}{N} \sum_i \langle z_i^2 \rangle  -(1-\theta)\langle q_{12} \rangle_2
\right)\\
\label{eqn:partialwA_G1RSB}
\frac{\partial }{\partial w}  \mathcal A^{(P)}_N =&  \frac{1}{2N} \sum_i \langle z_i^2 \rangle 
\end{align}
\end{lemma}
The proof is similar to Guerra's interpolation technique, since we omit it. 

\begin{proposition}
\label{prop:9}
The streaming of the 1-RSB interpolating quenched pressure obeys, at finite volume $N$, a standard transport equation, that reads as
\begin{align}
\label{eqn:transportequation1RSB_G}
\frac{d\mathcal A^{(P)}}{dt}&=\partial_t \mathcal A^{(P)}+\dot x^{(1)}\partial_{x_1} \mathcal A^{(P)} +\dot x^{(2)}\partial_{x_2} \mathcal A^{(P)} +\dot w \partial_{w} \mathcal A^{(P)} = S_{\textrm{1RSB}}(t, \bm r) + V_N(t, \bm r)
\end{align}
where
\begin{align}
\label{eqn:dotx1_G1RSB}
    \dot x^{(1)} &= -\frac{\beta'\,^2 P}{2} \q_1^{P-1} \\
    \dot x^{(2)} &= -\frac{\beta'\,^2 P}{2} (\q_2^{P-1} - \q_1^{P-1}) \label{eqn:dotx2_G1RSB}\\
    \dot w &= + \frac{\beta'\,^2 P}{2} \q_2^{P-1} \label{eqn:dotw_G1RSB}
\end{align}
and
\begin{align}
\label{eqn:f_G1RSB}
&S_{\textrm{1RSB}}(t, \bm r)  \coloneqq \frac{\beta'\,^2}{4} (P-1)\q_2^P-\frac{\beta'\,^2}{4} \theta (P-1)(\q_2^P - \q_1^P)\\
\label{eqn:V_G1RSB}
&V_N(t, \bm r) \coloneqq  \frac{\beta'\,^2}{4} (\theta-1) \left( \sum_{k=2}^P \langle \Delta q_{12}^k \rangle_2 \q_2^{P-k} \right)- \frac{\beta'\,^2}{4}\theta \sum_{k=2}^P \langle \Delta q_{12}^k \rangle_1 \q_1^{P-k} 
\end{align}
\normalsize
\end{proposition}

\begin{proof}
Keeping in mind the expression \eqref{eq:q12_RSB}, we start to evaluate explicitly $\dt \mathcal{A}^{(P)}_N$ by using (\ref{eqn:partialx1A_G1RSB} - \ref{eqn:partialwA_G1RSB}), so we get
\begin{align}
&\dt \mathcal{A}^{(P)}_N=  \frac{\beta'\,^2}{4} \Bigg\{  (\theta-1)\left[ \sum_{k=2}^P \langle \Delta q_{12}^k \rangle_2 \q_2^{P-k} +\bar q_2^{P}(1-P)+ P \q_2^{P-1} \langle q_{12} \rangle_2 \right] \notag \\
&- \theta \left[ \sum_{k=2}^P \langle \Delta q_{12}^k \rangle_1 \q_1^{P-k} +\bar q_1^{P} (1-P) +P \q_1^{P-1} \langle q_{12} \rangle_1 \right]\Bigg\} = \notag \\
&=V_N(t, \bm r) + S(t, \bm r)  + \frac{\beta'\,^2}{4}P (\theta-1) \q_2^{P-1} \langle q_{12} \rangle_2 -\theta(1-P)\q_1^P\notag \\
&= V_N(t, \bm r) +\frac{\beta'\,^2 (\theta-1)}{4}(1-P) \q_2^{P} - \frac{\beta'\,^2 \theta}{4}\q_1^{P} + \frac{\beta'\,^2 P}{2} \q_1^{P-1} \partial_{x^{(1)}} \mathcal{A}^{(P)}_N \notag \\
&+ \frac{\beta'\,^2 P}{2} (\q_2^{P-1}-\q_1^{P-1}) \partial_{x^{(2)}} \mathcal{A}^{(P)}_N - \frac{\beta'\,^2 P}{2}\q_2^{P-1} \partial_w \mathcal{A}^{(P)}_N
\end{align}
\normalsize
Thus, using \eqref{eqn:dotx1_G1RSB}-\eqref{eqn:dotw_G1RSB} we reach the thesis.

\end{proof}

\begin{remark} \label{r:above}
In the thermodynamic limit, in the 1RSB scenario, we have that the central moments greater than 2 tend to zero; in this way
\begin{equation} \label{eq:V0_HRSB_G}
\lim_{N \to \infty} V_N(t, \bm r) = 0.
\end{equation}
\end{remark}

Exploiting Remark \ref{r:above} we can prove the following
\begin{proposition} \label{propHRSB}
The transport equation associated to the interpolating pressure function $\mathcal A^{(P)}_N(t, \boldsymbol r)$, in the thermodynamic limit and under the 1RSB assumption, can be written as
\begin{align}
\label{eqn:solutionzeroV_G1RSB}
\partial_t \mathcal A^{(P)}+\dot x^{(1)}\partial_{x_1} \mathcal A^{(P)} +\dot x^{(2)}\partial_{x_2} \mathcal A^{(P)} +\dot w \partial_{w}  \mathcal A^{(P)} \notag \\
= \frac{\beta'\,^2}{4} (\theta-1) (1-P)\q_2^P - \frac{\beta'\,^2}{4}\theta (1-P) \q_1^P
\end{align}
whose explicit solution is given by
\begin{equation}
\small
    \begin{array}{lll}
         \mathcal A^{(P)}_{1RSB}= & t \left[ \dfrac{\beta'\,^2}{4}(P-1)\q_2^P-\dfrac{\beta'\,^2}{4} \theta(P-1)(\q_2^P - \q_1^P) \right] - \dfrac{1}{2} \log \left( 1-\lambda - w_0\right) 
         \\\\
    &+ \dfrac{1}{2\theta} \log \left(\dfrac{ 1-\lambda -w_0}{ 1-\lambda -w_0 - \theta x_0^{(2)}} \right)+ \dfrac{x_0^{(1)}}{2( 1-\lambda-w_0-  \theta x_0^{(2)})} 
    \end{array}
\end{equation}
\end{proposition}
\begin{proof}
We compute through characteristic method: 
\begin{align}
    \mathcal A^{(P)}_{\textrm{1RSB}}(t, \bm r)=\mathcal A^{(P)}_{\textrm{1RSB}}(0, \bm{r}-\bm{\dot r}t)+S_{\textrm{1RSB}}(t, \bm r )t
\end{align}
By putting (\ref{eqn:dotx1_G1RSB})-(\ref{eqn:dotw_G1RSB}) into (\ref{eqn:transportequation1RSB_G}) we find
\begin{equation}
\small
\label{eqn:solutionzeroV1_G1RSB}
\mathcal A^{(P)}_0(t, \bm r)=\mathcal A^{(P)}_0(0, \bm r_0) +t\left[ \dfrac{\beta'\,^2}{4}(P-1)\q_2^P-\dfrac{\beta'\,^2}{4} \theta(P-1)(\q_2^P - \q_1^P) \right]
\end{equation}
where $r_0$ can be obtained by using the equation of motion
\begin{equation}
\label{eqn:linearmotion_Gtrans}
\bm r = \bm r_0 + \dot{\bm r} t
\end{equation}
where the velocities are defined in (\ref{eqn:dotx1_G1RSB})-(\ref{eqn:dotw_G1RSB}). Then, all we have to compute is $\mathcal A^{(P)}_0(0, \bm r_0)$, that can be easily done because at $t=0$ the two body interaction vanishes and the (\ref{eqn:Z2_TrasP_G}) can be written as
\begin{align}
\label{eqn:A0fin_G1RSB}
\mathcal A^{(P)}_0(0, \bm r_0)=&-\frac{1}{2} \log \left( 1-\lambda - w_0\right) + \frac{1}{2\theta} \log \left( \frac{1-\lambda - w_0}{1-\lambda - w_0 - \theta x_0^{(2)}}\right) \notag \\
&+ \frac{x_0^{(1)}}{2(1-\lambda - w_0- \theta x_0^{(2)})}.
\end{align}
We omit the computation since it is similar to one-body term in Guerra's interpolating scheme, 1RSB assumption.
\newline
Then, putting together (\ref{eqn:solutionzeroV1_G1RSB})-(\ref{eqn:A0fin_G1RSB}) 
and (\ref{eqn:dotx1_G1RSB})-(\ref{eqn:dotw_G1RSB}), we finally achieve an explicit expression for the interpolating pressure of the Gaussian P-spin glass model in the 1RSB approximation.
\end{proof}

To sum up, we have the following main theorem for the 1RSB scenario.
\begin{theorem}
The 1-RSB quenched pressure for Gaussian P-spin glass model, in the thermodynamic limit, reads as
\begin{equation}
\label{eqn:hopfieldAfinal_G1RSB}
\begin{array}{lll}
      \mathcal A^{(P)}(\b,\lambda)= & \dfrac{\beta'\,^2}{4}(P-1)\q_2^P-\dfrac{\beta'\,^2}{4} \theta(P-1)(\q_2^P - \q_1^P)- \frac{1}{2} \log \left[ 1-\lambda + \beta'\,^2\frac{P}{2}\q_2^{P-1}\right]  \\\\
&+ \dfrac{1}{2\theta} \log \left[\dfrac{ 1-\lambda + \beta'\,^2\frac{P}{2}\q_2^{P-1}}{ 1-\lambda + \beta'\,^2\frac{P}{2}\q_2^{P-1} - \beta'\,^2\frac{P}{2} \theta \left(\q_2^{P-1} - \q_1^{P-1}\right)} \right]   \\\\
&+ \dfrac{\beta'\,^2 \frac{P}{4} \q_1^{P-1}}{ 1-\lambda+ \beta'\,^2\frac{P}{2}\q_2^{P-1} - \beta'\,^2\frac{P}{2}  \theta \left(\q_2^{P-1} - \q_1^{P-1}\right)}
\end{array}
\end{equation}

where the order parameters are ruled by 
\begin{equation}
    \begin{array}{lll}
         \bar{q}_1 = \dfrac{\beta'\,^2 \frac{P}{2} \q_1^{P-1}}{\left(1-\lambda + \beta'\,^2 \frac{P}{2} \q_2^{P-1} - \beta'\,^2\theta \frac{P}{2} (\q_2^{P-1}-\q_1^{P-1})\right)^2}
         \\\\
\bar{q}_2 = \q_1 + \dfrac{\beta'\,^2 \frac{P}{2} \left(\q_2^{P-1}-\q_1^{P-1}\right)}{\left(1-\lambda + \beta'\,^2 \frac{P}{2} \q_2^{P-1} - \beta'\,^2\theta \frac{P}{2} \left(\q_2^{P-1}-\q_1^{P-1}\right)\right)\left( 1-\lambda + \beta'\,^2 \frac{P}{2} \q_2^{P-1}\right)}
    \end{array}
    \label{eqn:self1rsbgauss_trans}
\end{equation}
\end{theorem}
\begin{proof}
By taking $\bm r= \boldsymbol 0$ and $t=1$ we find the P-spin Gaussian pressure in the 1RSB approximation. Extremizing \eqref{eqn:hopfieldAfinal_G1RSB} we obtain the self-consistency equations. 
\end{proof}
We stress that the expression in \eqref{eqn:hopfieldAfinal_G1RSB} is the same found with Guerra's interpolating scheme \eqref{A_1RSB_finalissima_G}.

\subsection{A deeper analysis of Gaussian spin glass model}
\label{sec:G_P2}
In this section we will focus more specifically on the case of $P=2$ for the Gaussian P-spin glass model under the assumption RS and 1-RSB, showing that we find the same results in \cite{AdrianoGauss}. 

Starting from the RS assumption, if we set $P=2$ in the quenched pressure of the Gaussian P-spin glass model \eqref{eq:pressure_GuerraRS_G}, we get 
\begin{equation}
\small
\begin{array}{lll}
     \mathcal{A}_{RS}(\b,\lambda) &=& \dfrac{\beta'\,^2\q}{2\left(1-\lambda+\beta'\,^2\q\right)}-\dfrac{1}{2}\ln{\left[1-\lambda+\beta'\,^2\q\right]}+\dfrac{\beta'\,^2}{4}\q^2\,.
\end{array}
\label{eq:pressure_G_RS_P2}
\end{equation}
Extremizing the statistical pressure in \eqref{eq:pressure_G_RS_P2} w.r.t.  the order parameter $\q$ we find the following 

\begin{corollary}
The minimum of the quenched statistical pressure in \eqref{eq:pressure_G_RS_P2} is achieved for
\begin{equation}
\label{eq:sol_RS_G_P2}
    \begin{array}{lll}
        \q=0 \ \ \ &\mathrm{if}&\b\leq 1-\lambda , 
         \\\\
         \q=\dfrac{\b-(1-\lambda)}{\beta'\,^2}\ \ \  &\mathrm{if}&\b>1-\lambda .
    \end{array}
\end{equation}
\end{corollary}
\begin{proof}
By setting equal to zero the derivative with respect to $\q$ of the eq. \eqref{eq:pressure_G_RS_P2}, we get
\begin{equation}
    \small
    \q=\dfrac{\beta'\,^2\q}{1-\lambda-\beta'\,^2\q}
\end{equation}
this equation give us two possible solution of the order parameter, namely
\begin{equation}
\small
    \begin{array}{lll}
         \q=0&\mathrm{and}& \q=\dfrac{\b-(1-\lambda)}{\beta'\,^2}
    \end{array}
    \label{eq:extrem_q}
\end{equation}
Now, computing the second derivative of the quenched pressure and replacing the values of $\q$ found in \eqref{eq:extrem_q}, we get
\begin{equation}
    \small
    \begin{array}{lll}
         \dfrac{\partial^2\mathcal{A}(\b,\lambda)}{\partial \q^2}\Big|_{\q=0}=\dfrac{\beta'\, ^2}{2\left(1-\lambda\right)^2}\left[\left(1-\lambda\right)^2-\beta'\, ^2\right], \\\\
         \dfrac{\partial^2\mathcal{A}(\b,\lambda)}{\partial \q^2}\Big|_{\q=(\b-1+\lambda)/\beta'\,^2}=\b  (\b -(1-\lambda)).
    \end{array}
\end{equation}
Thus, from the study of the sign of the previous equations we get the proof.
\end{proof}
Moving on the 1RSB case, if we set $P=2$ in \eqref{A_1RSB_finalissima_G}, we find
\begin{equation}
\small
\begin{array}{lll}
      A_{1RSB}(\b,\lambda)= & \dfrac{\beta'\,^2}{4} \q_2^2-\dfrac{\beta'\,^2}{4} \theta \left(\q_2^2 - \q_1^2\right) - \dfrac{1}{2} \log \left[ 1-\lambda + \beta'\,^2\q_2\right]  \\\\
&+ \dfrac{1}{2\theta} \log \left[\dfrac{ 1-\lambda + \beta'\,^2\q_2}{ 1-\lambda + \beta'\,^2\q_2 - \beta'\,^2 \theta \left(\q_2 - \q_1\right)} \right]   \\\\
&+ \dfrac{1}{2}\dfrac{ \beta'\,^2 \q_1}{ 1-\lambda+ \beta'\,^2\q_2 - \beta'\,^2  \theta \left(\q_2 - \q_1\right)}
\end{array}
\label{eq:pressure_G_1RSB_P2}
\end{equation}
Following the same steps presented for the RS assumption, we set to zero the derivatives of \eqref{eq:pressure_G_1RSB_P2} respect to the order parameters $\q_1$ and $\q_2$
\begin{equation}
\small
    \begin{array}{lll}
         \bar{q}_1 & =&\dfrac{\beta'\,^2 \q_1 }{\left(1-\lambda + \beta'\,^2  \q_2 - \beta'\,^2\theta (\q_2-\q_1)\right)^2}\,,
         \\\\
\bar{q}_2 -\q_1&= & \dfrac{\beta'\,^2 \left(\q_2-\q_1 \right)}{\left(1-\lambda + \beta'\,^2  \q_2 - \beta'\,^2\theta \left(\q_2-\q_1\right)\right)\left( 1-\lambda + \beta'\,^2  \q_2\right)}\,.
    \end{array}
    \label{eq:system_q1_q2_G}
\end{equation}
It is immediate to verify that the previous system of equations admits only the solution $\q_1 = \q_2=\q$, with $\q$ of the form presented in \eqref{eq:sol_RS_G_P2}. Therefore the solution in 1RSB approssimation  coincides with the one in RS approximation. As a matter of fact, the other solution of \eqref{eq:system_q1_q2_G} for $\q_1\neq\q_2$, reads as
\begin{equation}
\begin{array}{lll}
     (\q_1;\q_2)&=&\left(0\,;\dfrac{(\theta -2) (\lambda -1)-\sqrt{\theta ^2 (\lambda -1)^2-4 \beta ^2 (\theta -1)}}{2 \beta ^2 (\theta -1)}\right)\,,
\end{array}
\end{equation}
however, it can be shown that this solution is neither a minimum nor a maximum of \eqref{eq:pressure_G_1RSB_P2}, but a saddle point. 

Thus, we can conclude this section with the following theorem
\begin{theorem}
For the pairwise Gaussian spin glass model defined by the Hamiltonian
	\begin{equation}
	H_N(\boldsymbol z| \boldsymbol J) \coloneqq -\sqrt{\dfrac{1}{2N}}\SOMMA{i,j=1}{N,N}\,J_{ij}z_{i}z_{j}\;\;\;\mathrm{where}\;\;\;J_{ij}\sim\mathcal{N}(0,1)
	\end{equation}
the RS solution is exact.
\end{theorem}
\begin{proof}
The proof consists in showing that the 1RSB bound for the free energy gives the same result of the RS approximation.
\end{proof}


\section{Conclusions and outlooks} \label{conclusions}
In this paper we have generalized two rigorous mathematical methods for two different models, namely P-spin glass model and Gaussian P-spin model with $P>2$. We face up to RS and 1RSB assumptions, comparing an approach closer to Mathematical Physics, namely transport equation, and another linked to Statistics and Probability, namely Guerra's interpolating technique. We reach the same results via both techniques, proving that they are both valid for these models and we recovered the same expression in \cite{gardner, GrossMezard, Tala4, Tala5} for P-spin glass and \cite{crisanti, Tala3} for Gaussian P-spin glass models. 

Having the possibility to investigate the problem with mathematical tools from different fields could be powerful and could give us non-trivial result.

In addiction we showed that, as far as Gaussian spin glass model concerns, the replica symmetry expression for quenched statistical pressure is exact for $P=2$ case, in contrast to $P >2$ case. 

Further researches should be adressed to a deeper physical interpretation of this results, in particular on replica symmetry breaking assumptions. Another lines could be the rigorous mathematically confirmation of results in KRSB assumption, with K finite in \cite{GrossMezard}  and the approach of more challenging models, such as dense associative networks\cite{dense}. 

\appendix
\section{Proof of \eqref{potential_m1}-\eqref{potential_pq}}
\label{app:potenziali}
In this Appendix we show the computation of expression which is used either in Guerra's interpolating scheme or transport equations. 
Let's start with \eqref{potential_m1}. Using the notation $\Delta X= X- \bar{X}$ and exploiting Newton's binomial, we get 
\begin{equation}
    \begin{array}{lll}
         \l m ^{P}\r&=&\Big\l (m-\m+\m)^{P} \Big\r
         \\\\
         &=&\SOMMA{k=0}{P}\begin{pmatrix}P\\k
         \end{pmatrix}\m^{P-k}\Big\l(m-\m)^k\Big\r
         \\\\
         &=&(1-P)\m^P+P\m^{P-1}\l m\r+\SOMMA{k=2}{P}\begin{pmatrix}P\\k
         \end{pmatrix}\m^{P-k}\Big\l(m-\m)^k\Big\r\,.
    \end{array}
\end{equation}

In the same way (using $q_{12}$ instead of $m$ and $\q$ instead of $\m$) we prove \eqref{potential_pq}.

\section{Proof of Lemma \ref{lemma:tderRSB}}
\label{app:tder1RSB}
We prove the $t$ derivative of statistical pressure in 1RSB assumption. \small
\begin{align}
    &d_t \mathcal{A} =\frac{1}{N} \mathbb{E}_0\mathbb{E}_1\mathbb{E}_2 \mathcal{W}_2 \omega \left[ \dfrac{\b J_0 N}{2} m^P(\boldsymbol \sigma)-N\psi\,m(\boldsymbol \sigma)\right.\notag \\
    &\left.+\frac{\b J}{2\sqrt{t}}\sqrt{\dfrac{1 }{2N^{P-1}}}\SOMMA{i_1,\cdots,i_{_{P}}=1}{N,\cdots,N}z_{i_1\cdots,i_{_{P}}}\sigma_{i_1}\cdots\sigma_{i_{P}} -\frac{1}{2\sqrt{1-t}}\sum_{a=1}^2\left(A_a \SOMMA{i=1}{N} J_i^{(a)}\sigma_i\right)\right] = \notag \\
    &= - \psi \langle m \rangle + \frac{\b J_0}{2}\langle m^P \rangle + B_1 + B_2 + B_3, \\
    B_1 &=\sqrt{\frac{1}{2N^{P-1} t}}\frac{\b J}{2N} \mathbb{E}_0\mathbb{E}_1\mathbb{E}_2 \mathcal{W}_2 \omega\left( \SOMMA{i_1,\cdots,i_{_{P}}=1}{N,\cdots,N}z_{\bm i }\sigma_{i_1}\cdots\sigma_{i_{P}} \right) = \notag \\
    &=\sqrt{\frac{1}{2N^{P-1} t}}\frac{\b J}{2N}\SOMMA{i_1,\cdots,i_{_{P}}=1}{N,\cdots,N} \mathbb{E}_0\mathbb{E}_1\mathbb{E}_2 \partial_{z_{\bm i }} \left(\mathcal{W}_2 \omega\left( \sigma_{i_1}\cdots\sigma_{i_{P}} \right)\right) = \frac{\beta'\,^2 J^2 }{4} \left[ 1 + (\theta -1) \langle  q_{12}^{P} \rangle_2 - \theta \langle q_{12}^{P}\rangle_1 \right], \\
    B_2 &= -\frac{A_1}{2N\sqrt{1-t}}\sum_i \mathbb{E}_0\mathbb{E}_1\mathbb{E}_2 \mathcal{W}_2 \omega\left(z_i^{(1)}\sigma_i\right) = -\frac{A_1}{2N\sqrt{1-t}} \mathbb{E}_0\mathbb{E}_1\mathbb{E}_2 \partial_{z_i^{(1)}}\mathcal{W}_2 \omega\left(\sigma_i\right)  \notag \\
    &= \frac{A_1^2}{2} \left[ 1+ (\theta-1) \langle q_{12} \rangle_2 -\theta \langle q_{12} \rangle_1 \right], \\
    B_3 &= -\frac{A_2}{2N\sqrt{1-t}}\sum_i \mathbb{E}_0\mathbb{E}_1\mathbb{E}_2 \mathcal{W}_2 \omega\left(z_i^{(2)}\sigma_i\right) = -\frac{A_2}{2N\sqrt{1-t}} \mathbb{E}_0\mathbb{E}_1\mathbb{E}_2 \partial_{z_i^{(2)}}\mathcal{W}_2 \omega\left(\sigma_i\right)  \notag \\
    &= \frac{ A_2^2}{2} \left[ 1 + (\theta-1) \langle q_{12} \rangle_2 \right] .
\end{align}
\normalsize
Rearranging together we obtain the thesis.

\section*{Acknowledgments}
The authors are grateful to Adriano Barra for several fruitful discussions. \\
LA acknowledge Unisalento and INFN for partial financial support.\\
LA acknowledge partial financial support by the grant ``BULBUL" (bando di collaborazione industriale scientifica e tecnologica tra Italia ed Israele).

\end{document}